\documentclass[11pt]{article}

\usepackage[utf8]{inputenc}
\usepackage[margin=1in]{geometry}
\usepackage{mathtools, amsmath, amsthm, graphicx, enumitem, amsfonts, amssymb, verbatim, bbm, braket, xcolor, float, booktabs, multirow,colortbl}
\usepackage[linktocpage=true]{hyperref}
\usepackage[
backend=biber,
style=alphabetic,
maxbibnames=99,
]{biblatex}
\usepackage[capitalise,nameinlink]{cleveref}

\newcommand{\vect}[1]{\boldsymbol{#1}}
\newcommand{\av}{{\vect{a}}}
\newcommand{\bv}{{\vect{b}}}
\newcommand{\yv}{{\vect{y}}}
\newcommand{\zv}{{\vect{z}}}
\renewcommand{\v}{{\vect{v}}} % \v is the caron accent
\newcommand{\gammav}{{\vect{\gamma}}}
\newcommand{\Gammav}{\boldsymbol\Gamma}
\newcommand{\Thetav}{\boldsymbol\Theta}
\newcommand{\betav}{{\vect{\beta}}}
\newcommand{\paramv}{{\vect{\gamma}, \vect{\beta}}}
\newcommand{\of}[1]{\left( #1 \right)}
\newcommand{\ofc}[1]{\left\{ #1 \right\}}
\newcommand{\ofb}[1]{{\left[#1\right]}}

\newcommand{\oleft}[1]{\overset{\leftarrow}{#1}}
\newcommand{\oright}[1]{\overset{\rightarrow}{#1}}
\newcommand{\parent}{\mathfrak{p}}

\makeatletter
\newcommand{\Biggg}[1]{\mathopen{\bBigg@{3}#1}}
\makeatother
\makeatletter
\newcommand{\Bigggg}[1]{\mathopen{\bBigg@{4}#1}}
\makeatother
\makeatletter
\newcommand{\Biggggg}[1]{\mathopen{\bBigg@{5}#1}}
\makeatother
\makeatletter
\newcommand{\Bigggggg}[1]{\mathopen{\bBigg@{6}#1}}
\makeatother

\definecolor{lgreen}{rgb}{0.85,1,.85}

\allowdisplaybreaks

\newtheorem{lemma}{Lemma}
\newtheorem{theorem}{Theorem}
\newtheorem{corollary}{Corollary}
\newtheorem{claim}{Claim}
\newtheorem{remark}{Remark}

\hypersetup{
    colorlinks=true,
    linkcolor=blue,
    filecolor=magenta,      
    urlcolor=cyan,
    citecolor=violet,
}
\usepackage{complexity}
\newcommand{\defeq}{\stackrel{\mathrm{\scriptscriptstyle def}}{=}}

\begin{document}
\title{Performance of Variational Algorithms for Local Hamiltonian Problems on Random Regular Graphs}

\author{
  Kunal Marwaha \\
  University of Chicago
  \and
  Adrian She \\
  University of Toronto
  \and
  James Sud\thanks{Corresponding Author: \texttt{jsud@uchicago.edu}} \\
  University of Chicago
}

\date{}
\maketitle
\vspace{-1em}
\begin{abstract}
We design two variational algorithms to optimize specific 2-local Hamiltonians defined on graphs. 
Our algorithms are inspired by the Quantum Approximate Optimization Algorithm. 
We develop formulae to analyze the energy achieved by these algorithms with high probability over random regular graphs in the infinite-size limit, using techniques from~\cite{basso2022}. 
The complexity of evaluating these formulae scales exponentially with the number of layers of the algorithms, so our numerical evaluation is limited to a small constant number of layers. 
We compare these algorithms to simple classical approaches and a state-of-the-art worst-case algorithm.
We find that the symmetry inherent to these specific variational algorithms presents a major \emph{obstacle} to successfully optimizing the Quantum MaxCut (QMC) Hamiltonian on general graphs. 
Nonetheless, the algorithms outperform known methods to optimize the EPR Hamiltonian of~\cite{king2023} on random regular graphs, and the QMC Hamiltonian when the graphs are also bipartite. 
As a special case, we show that with just five layers of our algorithm, we can already prepare states within $1.62\%$ error of the ground state energy for QMC on an infinite 1D ring, corresponding to the antiferromagnetic Heisenberg spin chain.
\end{abstract}

\maketitle

\small
\tableofcontents

\normalsize

\section{Introduction}\label{sec:intro}
The natural quantum analogue of combinatorial optimization is the $\QMA$-complete local Hamiltonian problem~\cite{kempe2005}. 
Although we do not expect quantum algorithms to exactly solve every instance of this problem, we may instead design good quantum \emph{approximation} algorithms. 
So far, we do not know which efficient quantum algorithm will achieve the best performance on this problem.
This contrasts with the classical case~\cite{raghavendra2008optimal}, and is connected to the possibility of quantum PCPs~\cite{aharonov2013,hwang2022}.

One compelling formulation of the local Hamiltonian problem is inspired by the $\NP$-hard problem Maximum Cut. In Maximum Cut, one must 2-color the vertices of a graph $G$ in a way that maximizes the number of edges connecting different colors. In $\emph{Quantum MaxCut}$ (QMC), one aims to find a global quantum state that maximizes
the energy of a $2$-local Hamiltonian, where a single local term projects each edge of a graph $G$ into the antisymmetric singlet state.
This problem was previously studied as the spin-1/2 antiferromagnetic quantum Heisenberg model in statistical mechanics~\cite{heisenberg1928}, but has seen a resurgence in the quantum information literature~\cite{gharibian2019, anshu2020, parekh2021, parekh2022, lee2022, takahashi2023, king2023, lee2024, watts2024, kannan2024}.

Although most work on this topic considers worst-case approximation algorithms, it is natural to wonder whether there exist good \emph{average-case} approximation algorithms for these problems. Unlike the worst-case, \emph{random} instances of classical combinatorial optimization problems can often be solved in polynomial time~\cite{dembo2017,alaoui2020,alaoui2023,jekel2024}. 
Some random local Hamiltonians have near-optimal \emph{product states}~\cite{brandao2014, bergamaschi2024}, but it is unclear if better performance can be achieved with a quantum algorithm.

In this work, we consider a family of variational algorithms for Quantum MaxCut and related problems. 
These algorithms are inspired by the Quantum Approximate Optimization Algorithm (QAOA)~\cite{farhi2014, hogg2000}, itself inspired by quantum annealing~\cite{farhi2000a}. We leverage previous studies of the QAOA~\cite{basso2022, basso2022gamarnik} to analyze the expected performance of our algorithms on random regular graphs. 
Our results show that we can outperform previously-studied classical algorithms for some random Hamiltonians in the infinite size limit, including the EPR Hamiltonian of \cite{king2023} on random regular graphs and QMC on random regular bipartite graphs.

Despite this, our algorithms are not universally successful.
For QMC on more general graphs, we do not outperform even simple classical algorithms. 
This may be because the framework we use to analyze performance requires that our algorithms start with a symmetric (i.e. permutation-invariant) product state.
We find that this symmetry creates a \emph{barrier} to outperforming classical methods. However, since it is an artifact of our analysis, it is possible that similar variational algorithms may perform well on QMC when initialized with better product states. Finally, in the large-degree limit, we find that our algorithms do \emph{not} outperform classical approaches for any of Hamiltonians we consider, potentially due to the phenomenon known as \emph{monogamy of entanglement}.

\subsection{Our contributions}\label{sec:intro/contributions}
We consider local Hamiltonian problems defined on graphs. An instance $H$ is specified by a 2-local Hamiltonian term $H_{uv}$ and a graph $G$, where
\begin{align*}
    H \defeq \sum_{(u,v) \in E(G)} H_{uv}
\end{align*}
We consider three model Hamiltonians, each using a different 2-local term: \emph{QMC} (Quantum MaxCut), \emph{EPR}, and \emph{XY}. We define these models in \Cref{sec:background}.

We then introduce two variational algorithms to prepare high-energy states for these Hamiltonians. 
For each Hamiltonian, we derive the expected energy achieved by these algorithms on high-girth regular graphs as the number of qubits goes to infinity (\Cref{sec:iterations}).
We also find the limit of this value as the degree tends to infinity. These variational algorithms and analyses can be extended to arbitrary $k$-local Hamiltonian problems on (hyper)graphs; we do this in \Cref{apx:generic_ansatz}.

At this point, the expected energy is a function only of each algorithm's variational parameters. We numerically optimize these parameters for each Hamiltonian for degree $2 \le d \le 5$ and algorithm depth (number of variational layers) up to $5$.
In \Cref{sec:results}, we compare these values with simple classical algorithms described in \Cref{sec:background/classical_algos}.
On the EPR Hamiltonian, the variational algorithms outperform other algorithms at any depth.
Moreover, when the graphs are also bipartite, the variational algorithms outperform other algorithms on the QMC Hamiltonian.

When the graphs are further restricted to be \emph{edge-transitive}, we can compare to the state-of-the-art algorithms which use global solutions obtained by SDP solvers. 
We find that our algorithms match the energy of \cite{king2023}'s algorithm at just depth $1$, and can only improve at higher layers (\cref{sec:results/small_d}).
As an example, QMC on infinite bipartite $2$-regular graphs corresponds to the exactly-solvable Heisenberg XXX model on a 1D ring. For this case, our algorithm prepares a state within 1.62\% error of the true ground state, outperforming \cite{king2023}'s algorithm. 

Finally, we provide an explicit calculation the energy obtained by one of our algorithms at depth one for QMC, EPR, and XY on arbitrary unweighted graphs in \Cref{apx:mcansatz_onelayer}. An interactive notebook with all relevant code is available online~\cite{sud2024}.

\subsection{Related work}\label{sec:intro/related}
In the computer science literature, the first algorithms specifically for Quantum MaxCut (QMC) used semidefinite programs to prepare \emph{product states}~\cite{gharibian2019, parekh2021, parekh2022}.
It was then discovered that a short variational circuit~\cite{anshu2020, king2023, lee2022} or an algorithm based on maximum weight matchings~\cite{anshu2020, lee2024} could improve on this product state. \cite{king2023} introduced the EPR Hamiltonian, which is sign-problem free and thus not known to be $\QMA$-hard. As of the first version of this paper, the best worst-case approximation ratio is $0.584$ for QMC~\cite{lee2024} and $0.707$ for EPR~\cite{king2023}.
There exists a hierarchy of semidefinite programs to find the optimal energy of a QMC Hamiltonian, but  these are not expected to be efficient~\cite{takahashi2023,watts2024}.

Our algorithm is inspired by the Quantum Approximate Optimization Algorithm~\cite{farhi2014,hogg2000}, designed to solve problems in combinatorial optimization. Whether the QAOA actually gives a quantum computational advantage is a subject of much debate~\cite{farhi2015quantum,barak2015beating,chou2021limitations,basso2022gamarnik,boulebnane2024solving,chen2023local,montanaro2024quantum}. 
To analyze our algorithm on random local Hamiltonians, we generalize techniques developed to study the performance of the QAOA for the Maximum Cut problem on random graphs~\cite{basso2022, sureshbabu2024}. 
Other variational algorithms include quantum annealing~\cite{farhi2000a} and the Hamiltonian Variational Ansatz~\cite{wecker2015}.

The QMC Hamiltonian on a $2$-regular graph is exactly the 1-dimensional antiferromagnetic Heisenberg spin-$1/2$ chain.
This is solvable in the infinite size limit~\cite{bethe1931, faddeev1996}, with maximum per-edge energy $2\ln2$. Several recent works~\cite{dyke2021, li2022,sopena2022a, ruiz2024, ruiz2024b, raveh2024} design methods to prepare the ground state of this model on a quantum computer, but all approaches require circuit depth growing exponentially in the graph size.

We note that a concurrent paper \cite{kannan2024} simultaneously and independently studies the power of a QAOA-inspired variational circuit (``HamQAOA'') for Quantum MaxCut and related local Hamiltonian problems on random graphs. 
They are able to analyze the performance over a \emph{distribution} of arbitrary initial product states, bypassing the symmetry barriers in our variational algorithms.
As a result, HamQAOA may outperform our algorithm on the QMC Hamiltonian instances where our algorithm fails.
Additionally, HamQAOA achieves a slightly better error ($1.05\%$ vs $1.62\%$) on the infinite 1D Heisenberg XXX ring.
However, in the infinite size limit, the concurrent work explicitly computes the performance of HamQAOA only for QMC. 
We compute the performance of our algorithms for the QMC, EPR, and XY model Hamiltonians. We also show how to extend our analysis to evaluate the energy of ans\"atze consisting of (at most) $k$-local gates with respect to $k$-local Hamiltonians defined on regular $k$-uniform hypergraphs, for any constant $k$. Our analysis requires that each gate in the ansatz and each term in the Hamiltonian is invariant under permutations of the qubits it acts on non-trivially.

\section{Our setup}\label{sec:background}
We now introduce and motivate the problems and algorithms studied in this work.

\subsection{Problems}\label{sec:background/problems}

Maximum Cut (MaxCut) is a canonical combinatorial optimization problem in computer science. The (unweighted) problem is as follows: given an unweighted, undirected graph $G = (V, E)$ described by a set of $n$ vertices $V$ and a set of $m$ edges $E$, the goal is to partition the vertices into two disjoint sets $S$ and $\bar{S}$ such that the number of edges between the two sets is maximized.

Although MaxCut is known to be NP-hard, the Goemans-Williamson algorithm \cite{goemans1995} can be used to find \emph{approximate} solutions in polynomial time. Recently, the Quantum Approximate Optimization Algorithm \cite{farhi2014, hogg2000} has been studied for the MaxCut problem on random $D$-regular graphs, with performance monotonically increasing with the depth of the algorithm~\cite{basso2022}.

To solve this problem with a variational algorithm like the QAOA, one may formulate MaxCut as finding the maximal eigenstate of the below Hamiltonian (which we call the \emph{MC Hamiltonian}):
\begin{align}\label{eq:mc_H} \begin{split}
    H_{MC} 
    &\defeq \frac{1}{2} \sum_{(u,v) \in E}I_u I_v - Z_uZ_v,
    \\
    &= \sum_{(u,v) \in E} \ket{01}_{uv}\bra{01}_{uv} + \ket{10}_{uv}\bra{10}_{uv},
\end{split} \end{align}
where $I_u$ and $Z_u$ refer to the identity operator and the Pauli $Z$ operator acting solely on the qubit corresponding to vertex $u$. The second line demonstrates that the MC Hamiltonian is a sum over projectors on each edge into the subspace where qubit $i$ is not equal to qubit $j$ in the computational basis.

Quantum MaxCut (QMC) is a natural \emph{quantum} extension of this formulation of MaxCut. Here, the goal is to find the maximum eigenstate of what we call the $QMC$ \emph{Hamiltonian}:

\begin{align}\label{eq:qmc_H} \begin{split}
    H_{QMC} 
    &\defeq \frac{1}{2} \sum_{(u,v) \in E} I_u I_v - X_uX_v - Y_uY_v - Z_uZ_v, \\
    &= 2\sum_{(u,v) \in E} \ket{\psi^-}_{uv} \bra{\psi^-}_{uv},
\end{split} \end{align}

where the second line demonstrates that the QMC Hamiltonian is a sum over projectors on each edge onto the antisymmetric singlet state  $\ket{\psi^-}\defeq\frac{1}{\sqrt{2}} \ket{01}-\ket{10}$, with a coefficient of $2$. 

We also study two variants of Quantum MaxCut, where again the goal is to find a maximum eigenstate. First, we introduce the $XY$ \emph{Hamiltonian}:

\begin{align}\label{eq:xy_H}\begin{split}
    H_{XY} &\defeq \frac{1}{2} \sum_{(u,v) \in E} I_uI_v -Y_uY_v - Z_uZ_v \\
    &= \sum_{(u,v) \in E} \frac{1}{2} + \ket{\psi^-}_{uv} \bra{\psi^-}_{uv} - \ket{\phi^-}_{uv} \bra{\phi^-}_{uv}.
    \end{split}
\end{align}

The $XY$ Hamiltonian is is obtained by removing all $X_u X_v$ terms from the QMC Hamiltonian\footnote{The name ``XY'' is potentially confusing, as we only keep the $YY$ and $ZZ$ terms. 
We use this name because up to constant offsets and local rotations, this Hamiltonian is equivalent to the \emph{Heisenberg XY model} from statistical physics on arbitrary spin-coupling graphs. The choice of Hamiltonian in \cref{eq:xy_H} makes it easier to compare with \cref{eq:mc_H}, which allows us to borrow more notation from the analysis in \cite{basso2022}.}.
The second line shows that it is a sum over edges of a constant offset, a projector onto $\ket{\psi^-}$, and a projector onto $\ket{\phi^-}\defeq\frac{1}{\sqrt{2}} \ket{00}-\ket{11}$, with coefficients $1/2$, $1$, and $-1$, respectively.

Second, we study the $EPR$ \emph{Hamiltonian}, introduced in \cite{king2023}:
\begin{align}\label{eq:epr_H} \begin{split}
    H_{EPR} 
    &\defeq \frac{1}{2} \sum_{(u,v) \in E} I_u I_v + X_uX_v - Y_uY_v + Z_uZ_v \\
    &=  2\sum_{(u,v) \in E} \ket{\psi^+}_{uv}\bra{\psi^+}_{uv},
\end{split} \end{align}
where second line demonstrates that EPR is a sum over projectors on each edge onto the EPR pair $\ket{\psi^+}\defeq\frac{1}{\sqrt{2}} (\ket{00}+\ket{11})$ with coefficient $2$. Thus, the all-zeros state $\ket{0}^{\otimes n}$ is an optimal product state for every graph.
As a consequence, the EPR Hamiltonian removes all of the ``classical'' difficulty from the problem, as the optimization problem purely asks how best to distribute entanglement among vertices. 

\begin{remark}[\cite{king2023}]\label{rem:bipartite_qmc_epr_equivalence}
    On bipartite graphs with a vertex partition $(A,B)$, the EPR and QMC Hamiltonians are equivalent under a local Pauli $Y$ rotation of all qubits in (wlog) part $A$.
\end{remark}

These three Hamiltonians are all of the form

\begin{align}\label{eq:2local_H_form}
    H = \sum_{(u,v) \in E}  c_I I_u I_v + c_X X_uX_v + c_Y Y_uY_v + c_Z Z_uZ_v,
\end{align}

and we in fact derive the performance of our algorithm for any $\{c_I, c_X, c_Y, c_Z\}\in \mathbb{R}$.

\subsection{Quantum algorithms}\label{sec:background/quantum_algos}
We analyze two variational algorithms 
to optimize the Hamiltonians defined above.
Both are inspired by the QAOA, which has a \emph{mixer} Hamiltonian and a \emph{phaser} or \emph{cost} Hamiltonian. 
In a classical problem, the phaser is exactly the Hamiltonian we wish to optimize. 
This does not easily extend to local Hamiltonian problems, as implementing $e^{i \gamma H}$ is nontrivial when $H$ contains non-commuting terms. To overcome this, we choose the phaser to be some other classical Hamiltonian that drives the QAOA.

The first algorithm is exactly the QAOA driven by the \emph{MaxCut Hamiltonian}, preparing a state $\ket{\gammav, \betav}$. We refer to this algorithm as the \emph{MC ansatz}:
\begin{align}\begin{split}\label{eq:mc_ansatz}
    &\ket{\gammav, \betav} =e^{-i \beta_p B}   e^{-i \gamma_p C_z}  \cdots e^{-i \beta_1 B} e^{-i \gamma_1 C_z} \ket{s}, \\
    &\ket{s} \defeq \ket{+}^{\otimes n}, \quad B \defeq \sum_{j=1}^n X_j, \quad C_z \defeq -\frac{1}{\sqrt{D}}\sum_{(u,v)\in E}Z_uZ_v, \\
    &\gammav \defeq (\gamma_1, \dots, \gamma_p), \quad \betav = (\beta_1, \dots, \beta_p).
\end{split}\end{align}
In the phasing operator $e^{i\gamma C_z}$, we apply a scaled and shifted version of the MaxCut Hamiltonian. 
We do this for ease of analysis: it recovers
the original QAOA operator under a rescaling of $\gammav$.

We refer to our second variational algorithm as the \emph{XY ansatz}, preparing the state $\ket{\gammav_y, \gammav_z,  \betav}$. This is the QAOA driven by a first-order Trotter expansion of the XY Hamiltonian:
\begin{align}
\ket{\gammav_y, \gammav_z,  \betav} =e^{-i \beta_p B}   e^{-i \gamma_{y_p} C_y}  e^{-i \gamma_{z_p} C_z}  \cdots e^{-i \beta_1 B} e^{-i \gamma_{y_1} C_y}  e^{-i \gamma_{z_1} C_z}  \ket{s}\label{eq:xy_ansatz}.
\end{align}
Here, $\gammav_z$, $\gammav_y$ are defined the same way as $\gammav$, with suitable subscripts $y$ and $z$. Also, $C_y$ is defined as $C_z$ above, but with all Pauli $Z_j$ replaced by $Y_j$. 

This algorithm is similar in spirit to the Hamiltonian Variational Ansatz (HVA) \cite{wecker2015}, which iteratively applies commuting sets of terms in the Hamiltonian. However, unlike in the HVA, we do not start in an eigenstate of one of the sets of commuting terms.

\subsection{Classical algorithms}\label{sec:background/classical_algos}
We study three simple classical algorithms to benchmark the quantum algorithms introduced above. 
\begin{itemize}
    \item ZERO: Prepare all qubits in the $\ket{0}$ state.
    \item MATCH: Find a maximum weight matching in $G$. For each edge in the matching, prepare the state $\ket{\psi^-}$ on the associated qubits for $H_{QMC}$ and $H_{XY}$, and $\ket{\phi^+}$ on the associated qubits for the $H_{EPR}$. Prepare the remaining qubits to the maximally mixed state.
    \item CUT: Find the maximum cut of $G$, and return the associated computational basis state.
\end{itemize}
The algorithm ZERO clearly runs in polynomial time. MATCH runs in polynomial time due to efficient algorithms for maximum weight matching \cite{edmonds1965}. Although CUT does not necessarily run in polynomial time, it upper-bounds any polynomial-time strategy that efficiently finds a large cut of $G$.
For every unweighted graph $G$ with $m$ edges, maximum matching size $M$, and maximum cut size $C$, the algorithms achieve the energies listed in \Cref{tab:classical_algos_regular_graphs}.
\begin{table}[ht]
\begin{equation*}
    \begin{array}{|c||c|c|c|}
    \hline
        & \text{ZERO}  & \text{MATCH} & \text{CUT} \\
    \hline
    \text{QMC} & 0   & (3M+m)/2 & C \\
    \text{EPR} & m &   (3M+m)/2 & C \\
    \text{XY}  & 0  & (2M+m)/2 & C \\
    \hline
\end{array} 
\end{equation*}
\caption{\footnotesize Energies obtained by the classical algorithms ZERO, MATCH, and CUT on the Hamiltonians QMC, EPR, and XY, defined on unweighted graphs.}
\label{tab:classical_algos_regular_graphs}
\end{table}

In some cases, we also compare the quantum algorithms to the worst-case algorithm of \cite{king2023}. This algorithm finds a \emph{global} solution using a semidefinite program, and locally optimizes with a short variational circuit.
Because of this global structure, we cannot easily predict the performance of this algorithm from the locally tree-like properties of random regular graphs.
However, when a graph is \emph{edge-transitive} (i.e. there exists an automorphism of the graph mapping any edge to any other edge) the output of the semidefinite program can be written down exactly. For example, this occurs for the infinite 1D ring and for infinite regular trees.
We refer to this algorithm as KING.

Throughout this work, we also use the terms \{ZERO, MATCH, CUT, KING\} to refer to the energy obtained by the corresponding algorithms on a particular problem. The problem will be clear from context.

\section{Iterations for the quantum ans\"atze}\label{sec:iterations}
We now give formulae that compute the energy obtained by the MC and XY ansatz at depth $p$ for Hamiltonians of the form in \cref{eq:2local_H_form} with high probability over random $(D+1)$-regular graphs with high-girth ($\geq 2p+1$), both for small degree $D$ and in the $D \rightarrow \infty$ limit. We define random regular graphs to be those chosen uniformly 
at random space of all labeled regular graphs \cite{wormald1999}. The formulae are largely inspired by the work of \cite{basso2022}. We further generalize the analysis in this section to more expressive ans\"atze and for larger locality Hamiltonians in \cref{apx:generic_ansatz}.

We comment briefly about why analyzing a high-girth graph is asymptotically equivalent to analyzing a random regular graph. For any constant $p$, a random regular graph has with high probability a constant number of cycles up to any size $2p+2$~\cite{WORMALD1981168}. As a result, with high probability, all but a vanishing fraction of depth-$p$ neighborhoods are locally treelike. Finally, the energy our algorithms achieve on each edge is only determined by the local neighborhood around each edge.

We denote the energy obtained by a specific ansatz A with optimal parameters $\Thetav$ at depth $p$ on a specific Hamiltonian H with $\mathrm{A}_{(p, \mathrm{H})}$. For example, $\mathrm{MC}_{(p, \mathrm{QMC})}$ denotes the  energy obtained by the MC ansatz on the QMC Hamiltonian at optimal parameters. When the Hamiltonian can be inferred from context, we simply denote the energy as (for the same example) $\mathrm{MC}_p$. There is a subtle distinction regarding the phrase ``at optimal parameters''. In the present work we obtain parameters through numerical optimization, without a certificate of global optima. Thus, we actually obtain lower bounds for $\mathrm{A}_{(p, \mathrm{H})}$. 

Following \cite{basso2022}, it is convenient to also denote the expected per-edge energy (or \emph{normalized energy}) for some state described by parameters $\Thetav$ and ansatz $A$ on Hamiltonian $H$ by

\begin{align}
	\frac{\braket{\Thetav|H|\Thetav}}{m} = E_0 + \frac{\nu_{p}(D,H, A,\Thetav)}{\sqrt{D}}, \label{eq:nu_DT}
\end{align} 

where $E_0$ is the expected normalized energy of a trivial randomized algorithm (for example, a maximally mixed state) on a single edge, given by the coefficient in front of the identity terms, and $\nu_{p}(D,H, A,\Thetav)$ is a constant that represents the expected normalized energy obtained by ansatz $A$ with some set of parameters $\Thetav$ on $H$ for $(D+1)$-regular graphs.

Exactly as in \cite{basso2022}, due to the regularity and high-girth of the graph, all distance-$p$ local neighborhoods of the graph are simply two $D$-ary trees, glued at the roots. Since QAOA at depth $p$ only sees these neighborhoods, the contribution to the expected normalized energy is the same for each neighborhood. Thus, we can write express $\nu_{p}(D,H, A,\Thetav)$ in terms of the non-identity components of the the expectation value of a single edge in the graph. For example, on the XY Hamiltonian we have, 
\begin{align}
     \nu_p(D, XY, A, \Thetav)= -\frac{\sqrt{D}}{2}\bra{\Thetav} X_u X_v + Y_uY_v \ket{\Thetav},\label{eq:nu_XY_DT}
\end{align}

and for a Hamiltonian $H$ of the form in \cref{eq:2local_H_form} we have,
\begin{align}
    \nu_p(D, H, A, \Thetav) = \frac{\sqrt{D}}{2}\bra{\Thetav} c_X X_u X_v + c_Y Y_uY_v + c_Z Z_uZ_v \ket{\Thetav} ,\label{eq:nu_H_DT}
\end{align}

for any edge $(u,v)$. We then can define the large-degree limit of the normalized energy:
\begin{align*}
    \nu_p(H, A, \Thetav) \defeq \lim_{D \rightarrow \infty} \nu_p(D, H, A, \Thetav).
\end{align*}
We use $\nu_p(H,A)$ to denote the maximal value of $\nu_p (H, A, \Thetav)$ over all parameters $\Thetav$:
\begin{align}\label{eq:nu}
	\nu_p(H, A) \defeq  \max_{\Thetav} \; \nu_p (H, A, \Thetav).
\end{align} 

\subsection{MC ansatz}\label{sec:iterations/mc_ansatz}

We first give the energy obtained by the MC ansatz.
We defer the proof to \cref{apx:iteration_proofs/mc}. 
\begin{theorem}[Energy of the MC ansatz]\label{thm:mc_ansatz_energy}
The energy obtained by the MC ansatz at depth $p$ with parameters $(\gammav, \betav)$ on $(D+1)$-regular graphs with girth $>2p+1$ on a $2$-local Hamiltonian in the form of \cref{eq:2local_H_form} is given by the iteration

\begin{align*}
    &\braket{\paramv | H | \paramv}=c_I m + c_X \braket{\paramv | X_L X_R | \paramv} + c_Y \braket{\paramv | Y_L Y_R | \paramv} + c_Z \braket{\paramv | Z_L Z_R | \paramv}, \\
    &\braket{\paramv | X_L X_R | \paramv} = \sum_{\av, \bv} f'(\av) f'(\bv) H^{(p)}_D(\av) H^{(p)}_D(\bv) \cos\left[\frac{\Gammav \cdot (\av \bv)}{\sqrt{D}}\right], \\
    &\braket{\paramv | Y_L Y_R | \paramv} = i\sum_{\av, \bv} a_0 b_0 f'(\av) f'(\bv) H^{(p)}_D(\av) H^{(p)}_D(\bv) \sin\left[\frac{\Gammav \cdot (\av \bv)}{\sqrt{D}}\right], \\
    &\braket{\paramv | Z_L Z_R | \paramv} = -i \sum_{\av, \bv} a_0 b_0 f(\av) f(\bv) H^{(p)}_D(\av) H^{(p)}_D(\bv) \sin\left[\frac{\Gammav \cdot (\av \bv)}{\sqrt{D}}\right],
\end{align*}

where $\av$, $\bv$ are $\pm 1$-valued, length-$2p+1$ bitstrings indexed as $(a_1, a_2, \ldots, a_p, a_0, a_{-p}, \ldots, a_{-2}, a_{-1})$ and $\av \bv$ denotes element-wise multiplication. In the above we define

\begin{align*}
    \Gammav &\defeq (\gamma_1, \gamma_2, \ldots, \gamma_p, 0, -\gamma_p, \ldots, -\gamma_2, -\gamma_1), \\
    f(\av) &\defeq \frac{1}{2} \braket{a_1 | e^{i \beta_1 X} | a_2} \cdots \braket{ a_{p-1} | e^{i \beta_{p-1} X} | a_p} \braket{a_p | e^{i \beta_p X} | a_0} \nonumber \\
    &\quad \times \braket{a_{0} | e^{-i \beta_p X} | a_{-p}} \braket{a_{-p} | e^{-i \beta_{p-1} X} | a_{-(p-1)}} \cdots \braket{a_{-2} | e^{-i \beta_1 X} | a_{-1}}, \\
    f'(\av) &\defeq \frac{1}{2} \braket{a_1 | e^{i \beta_1 X} | a_2} \cdots \braket{ a_{p-1} | e^{i \beta_{p-1} X} | a_p} \braket{a_p | e^{i \beta_p X} | -a_0} \nonumber \\
    &\quad \times \braket{a_{0} | e^{-i \beta_p X} | a_{-p}} \braket{a_{-p} | e^{-i \beta_{p-1} X} | a_{-(p-1)}} \cdots \braket{a_{-2} | e^{-i \beta_1 X} | a_{-1}}, \\
    H_D^{(m)}(\av) &\defeq \of{\sum_{\bv}  H^{(m-1)}_D(\bv) \exp\ofb{- \frac{i}{\sqrt{D}} \Gammav  \cdot (\av \bv)} f(\bv)}^D \nonumber \\
    &= \of{\sum_{\bv}  H^{(m-1)}_D(\bv) \cos\ofb{\frac{1}{\sqrt{D}} \Gammav  \cdot (\av \bv)} f(\bv)}^D , \\
    H_D^{(0)}(\av) &\defeq 1,
\end{align*}
and where $\ket{a_i}$ represents the single qubit state with eigenvalue $a_i$ in the computational (Pauli $Z$) basis.
The time and space complexity of evaluating this iteration are $O(p 16^p)$ and $O(16^p)$, respectively.    
\end{theorem}
We also take the $D \to \infty$ limit of the energy in the following corollary.
\begin{corollary}[Infinite degree limit of MC ansatz energy]\label{cor:mc_inf_d_energy}
    As $D\rightarrow \infty$, the expected normalized energy obtained by the MC ansatz at depth $p$ with parameters $(\gammav, \betav)$ of order $\Theta(1)$ on $(D+1)$-regular graphs with girth $>2p+1$ on the QMC, EPR, and XY Hamiltonians satisfy the following:
    \begin{enumerate}
        \item $\braket{\gammav, \betav | H_{EPR}|\gammav, \betav} \leq ZERO$.
        \item $\nu_p(QMC, MC, \gammav, \betav) \leq \nu_p(XY, MC, \gammav, \betav)$.
        \item $\nu_p(XY, MC, \gammav, \betav) 
        = \frac{i}{2} \sum_{j=-p}^{p} \Gamma_j  \of{G_{0,j}^{(m)}}^2- \of{G_{0,j}^{'(m)}}^2$.
    \end{enumerate}

    where for the last item we defined
    \begin{align*}
        H^{(p)}(\av) &\defeq \exp\ofb{-\frac{1}{2}\sum_{j,k=-p}^p \Gamma_j \Gamma_k a_j a_k G_{j,k}^{(p-1)}}, \\
        G_{j,k}^{(m)} &\defeq \sum_{\av} f(\av) a_ja_k\exp\ofb{-\frac{1}{2}\sum_{j',k'} \Gamma_{j'} \Gamma_{k'} a_{j'} a_{k'} G_{j',k'}^{(m-1)}}, \\
        G_{j,k}^{'(m)} &\defeq \sum_{\av} f'(\av) a_ja_k\exp\ofb{-\frac{1}{2}\sum_{j',k'} \Gamma_{j'} \Gamma_{k'} a_{j'} a_{k'} G_{j',k'}^{(m-1)}}, \\ 
        G_{j,k}^{(1)}&\defeq G_{j,k}^{'(1)}\defeq 1.
    \end{align*}
Here, $\av$, $f(\av)$, and $\Gammav$ are defined as in \cref{thm:mc_ansatz_energy}. Finally, the iteration can be evaluated in $\mathcal{O}(p^2 4^p)$ time and $\mathcal{O}(4^p)$ space or in $\mathcal{O}(p^4 4^p)$ time and $\mathcal{O}(p^2)$ space.
\end{corollary}
In particular, \cref{cor:mc_inf_d_energy} demonstrates that 1) the energy of the MC ansatz on the EPR Hamiltonian with any angles is upper bounded by the ZERO algorithm in the large degree limit and 2) the normalized energy of the MC ansatz on the QMC Hamiltonian with any angles is upper bounded by the normalized energy of the MC ansatz on the XY Hamiltonian with the same angles.
The proof of \cref{cor:mc_inf_d_energy} is deferred to \cref{apx:iteration_proofs/mc}.
We optimize over angles in \cref{sec:results/infinite_d} to determine the optimal normalized energy for the XY Hamiltonian.

\subsection{XY ansatz}\label{sec:iterations/xy_ansatz}
For the XY ansatz, we may similarly evaluate the energy, deferring the proof to \cref{apx:iteration_proofs/xy}:
\begin{theorem}[Energy of the XY ansatz]\label{thm:xy_ansatz_energy}
The energy obtained by the XY ansatz at depth $p$ with parameters $\Thetav = (\gammav_y, \gammav_z, \betav)$ on $(D+1)$-regular graphs with girth $>2p+1$ on a $2$-local Hamiltonian in the form of \cref{eq:2local_H_form} is given by the iteration

\begin{align*}
    &\braket{\Thetav | H | \Thetav}=c_I m + c_X \braket{\Thetav | X_L X_R | \Thetav} + c_Y \braket{\Thetav | Y_L Y_R | \Thetav} + c_Z \braket{\Thetav | Z_L Z_R | \Thetav}, \\
    &\braket{\Thetav | X_L X_R | \Thetav} = \sum_{\av, \bv} f'(\av) f'(\bv) H^{(p)}_D(\av) H^{(p)}_D(\bv) \cos\left[\frac{\Gammav \cdot (\av \bv)}{\sqrt{D}}\right], \\
    &\braket{\Thetav | Y_L Y_R | \Thetav} = i\sum_{\av, \bv} a_0 b_0 f'(\av) f'(\bv) H^{(p)}_D(\av) H^{(p)}_D(\bv) \sin\left[\frac{\Gammav \cdot (\av \bv)}{\sqrt{D}}\right], \\
    &\braket{\Thetav | Z_L Z_R | \Thetav} = -i \sum_{\av, \bv} a_0 b_0 f(\av) f(\bv) H^{(p)}_D(\av) H^{(p)}_D(\bv) \sin\left[\frac{\Gammav \cdot (\av \bv)}{\sqrt{D}}\right],
\end{align*}

where $\av$, $\bv$ are $\pm 1$-valued, length-$4p+1$ bitstrings indexed as 

\begin{align}
    ({a_z}_1, {a_y}_1, {a_z}_2, {a_y}_2 \ldots, {a_z}_p, {a_y}_p, a_0, {a_y}_{-p}, {a_z}_{-p} \ldots, {a_y}_{-2}, {a_z}_{-2}, {a_y}_{-1}, {a_z}_{-1}),
\end{align}

and where we define

\begin{align*}
    \Gammav &\defeq ({\gamma_z}_1, {\gamma_y}_1, {\gamma_z}_2, {\gamma_y}_2, \ldots, {\gamma_z}_p, {\gamma_y}_p, 0, -{\gamma_y}_p, -{\gamma_z}_p, \ldots, -{\gamma_y}_1, -{\gamma_z}_1), \\
    f(\av) \defeq& \frac{1}{2} \braket{a_{z_1}|a_{y_1}} \braket{a_{y_1}|e^{i \beta_1 X_u}|a_{z_2}} \cdots \braket{a_{z_p}|a_{y_p}} \braket{a_{y_p}|e^{i \beta_p X_u}|a_{0}} \nonumber\\
    &\times \braket{a_{0}|e^{-i \beta_p X_u}|a_{y_{-p}}}\braket{a_{y_{-p}}|a_{z_{-p}}} \cdots \braket{a_{z_{-2}}|e^{-i \beta_1 X_u}|a_{y_{-1}}}\braket{a_{y_{-1}}|a_{z_{-1}}}, \\
    f'(\av) \defeq& \braket{a_{z_1}|a_{y_1}} \braket{a_{y_1}|e^{i \beta_1 X_u}|a_{z_2}} \cdots \braket{a_{z_p}|a_{y_p}} \braket{a_{y_p}|e^{i \beta_p X_u}|-a_{0}} \\
    &\times \braket{a_{0}|e^{-i \beta_p X_u}|a_{y_{-p}}}\braket{a_{y_{-p}}|a_{z_{-p}}} \cdots \braket{a_{z_{-2}}|e^{-i \beta_1 X_u}|a_{y_{-1}}}\braket{a_{y_{-1}}|a_{z_{-1}}},  \\
    H_D^{(m)}(\av) &\defeq \of{\sum_{\bv}  H^{(m-1)}_D(\bv) \exp\ofb{- \frac{i}{\sqrt{D}} \Gammav  \cdot (\av \bv)} f(\bv)}^D \nonumber \\
    &= \of{\sum_{\bv}  H^{(m-1)}_D(\bv) \cos\ofb{\frac{1}{\sqrt{D}} \Gammav  \cdot (\av \bv)} f(\bv)}^D  \\
    H_D^{(0)}(\av) &\defeq 1,
\end{align*}

and where $\ket{{a_O}_i}$ represents the single qubit state with eigenvalue ${a_O}_i$ in the basis $O \in \{Y,Z\}$. The time and space complexity of evaluating this iteration are $O(p \cdot 64^p)$ and $O(64^p)$, respectively.
\end{theorem}
Much like the MC ansatz, we can take the $D \rightarrow \infty$ limit to get the following corollary.
\begin{corollary}[Infinite degree limit of MC ansatz energy]\label{cor:xy_inf_d_energy}
    As $D\rightarrow \infty$, the expected normalized energy obtained by the XY ansatz at depth $p$ with parameters $\Thetav = (\gammav_y, \gammav_z, \betav)$ of order $\Theta(1)$ on $(D+1)$-regular graphs with girth $>2p+1$ on the QMC, EPR, and XY Hamiltonians satisfy the following:
    \begin{enumerate}
        \item $\braket{\Thetav | H_{EPR}|\Thetav} \leq ZERO$.
        \item $\nu_p(QMC, XY, \Thetav) \leq \nu_p(XY, XY, \Thetav)$.
        \item $\nu_p(XY, XY, \Thetav) 
        = \frac{i}{2} \sum_{j=-p}^{p} \Gamma_j  \of{G_{0,j}^{(m)}}^2- \of{G_{0,j}^{'(m)}}^2$.
    \end{enumerate}

where in the last line we define
    \begin{align}
        H^{(p)}(\av) &\defeq \exp\ofb{-\frac{1}{2}\sum_{j,k=-p}^p \Gamma_j \Gamma_k a_j a_k G_{j,k}^{(p-1)}}, \\
        G_{j,k}^{(m)} &\defeq \sum_{\av} f(\av) a_ja_k\exp\ofb{-\frac{1}{2}\sum_{j',k'} \Gamma_{j'} \Gamma_{k'} a_{j'} a_{k'} G_{j',k'}^{(m-1)}}, \\
        G_{j,k}^{'(m)} &\defeq \sum_{\av} f'(\av) a_ja_k\exp\ofb{-\frac{1}{2}\sum_{j',k'} \Gamma_{j'} \Gamma_{k'} a_{j'} a_{k'} G_{j',k'}^{(m-1)}}, \\ 
        G_{j,k}^{(1)}&\defeq G_{j,k}^{'(1)}\defeq 1.
    \end{align}
Here, $\av$, $f(\av)$, and $\Gammav$ are defined as in \cref{thm:xy_ansatz_energy}. The iteration can be evaluated in $\mathcal{O}(p^5 16^p)$ time and $\mathcal{O}(p^2)$ space.
\end{corollary}
In particular, \cref{cor:xy_inf_d_energy} demonstrates that 1) the energy of the XY ansatz on the EPR Hamiltonian with any angles is upper bounded by the ZERO algorithm in the large degree limit and 2) the normalized energy of the XY ansatz on the QMC Hamiltonian with any angles is upper bounded by the normalized energy of the XY ansatz on the XY Hamiltonian with the same angles.
The proof of \cref{cor:xy_inf_d_energy} is deferred to \cref{apx:iteration_proofs/xy}.
We optimize over angles in \cref{sec:results/infinite_d} to determine the optimal normalized energy for the XY Hamiltonian.

\section{Results}\label{sec:results}
We benchmark the performance of our variational algorithms (MC ansatz and XY ansatz) against simple classical approaches and prior work.
We focus on both $(D+1)$-regular graphs for $2 \le D+1 \le 5$ and in the $D \to \infty$ limit.
We summarize our results in \cref{fig:optimal_algs}, which depicts the best algorithm at each value of $d$ and for each of the three problems we consider.

We investigate small-degree graphs in \Cref{sec:results/small_d} and the infinite-degree limit in \Cref{sec:results/infinite_d}. Before going into more detail, we comment on how we arrived at these numerical values.
\begin{figure}
    \centering
    \includegraphics[width=0.85\linewidth]{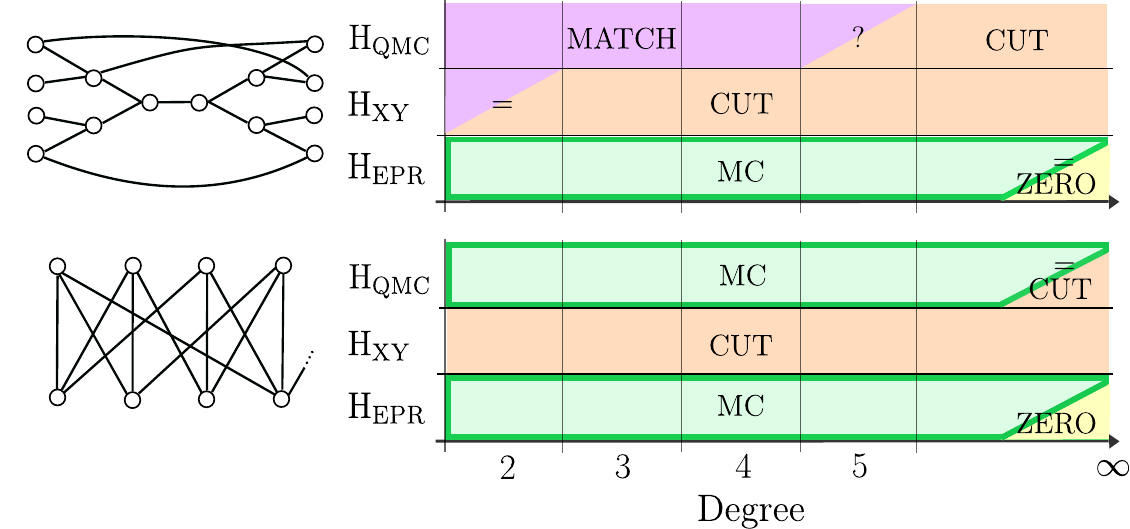}
    \caption{\footnotesize Optimal algorithms for the QMC, EPR, and XY Hamiltonians on random regular graphs (top) and random regular bipartite graphs (bottom). The \emph{green boxed regions} indicate regimes where our \emph{quantum algorithms outperform the classical algorithms} ZERO, MATCH, and CUT. For the MC ansatz we use $p=5$ when degree $\in [2,5]$, and $p=10$ at infinite degree. For the XY ansatz we use $p=2$ when degree $\in [2,5]$, and $p=4$ at infinite degree.}
    \label{fig:optimal_algs}
\end{figure}
\paragraph{Variational algorithms}
The iterations in \Cref{sec:iterations} depend on parameters $\Thetav$. Here, we numerically optimize over all choices of parameters. To achieve this, we implemented the iterations in Python, using both CPU and CuPy (for GPU implementation). Parameters were optimized by selecting a constant number of initial points (typically $20-50$) uniformly at random from range $[-1,1]$ for each parameter, optimizing with the BFGS algorithm for each trial, then choosing the best energy among trials. The code to reproduce these results is available at \cite{sud2024}. For the MC ansatz at $p=1$, we also derive the expected energy in closed form in \Cref{apx:mcansatz_onelayer}. In fact, in \cref{eq:p1_xyz_exp_vals_analytic}, we compute the \emph{exact} energy on \emph{any} unweighted graph for each Hamiltonian.

\begin{remark}[Last mixing layer commutes with QMC and XY Hamiltonians]\label{rem:mixer_commute}
    The last mixing layer $e^{-i\beta_p B}$ commutes with the XY and QMC Hamiltonian, and thus has no effect on the energy. Because of this, we leave out the last parameter $\beta_p$ or set $\beta_p=0$ when presenting results for these cases. The mixing layer does not commute with the EPR Hamiltonian.
\end{remark}

To verify the accuracy of our code and analysis of the MC and XY ans\"atze, we compared our calculations to that of exact statevector simulations on regular subgraphs of equal depth. This was used to cross-check the MC ansatz at depth $2$ for $D=2$ and depth $1-4$ for $D=1$. The XY ansatz was cross-checked by setting all $\gamma_y$ to $0$ and comparing to the MC ansatz at the same depth. This was done at depths $1$, $2$, and $3$ for a set of random parameters at infinite $D$ and at depth $1$ for a set of random parameters for $D=3$. We also verified that the depth-1 results for the MC ansatz agreed with the analytically-derived formula for general graphs given in \cref{eq:p1_reg_triangle_free_xyz_exp_vals_analytic}.

\paragraph{Classical algorithms}
For each local Hamiltonian problem, we compare the ans\"atze to the three simple classical algorithms ZERO, MATCH, and CUT (as described in \cref{sec:background/classical_algos}). 
On a $(D+1)$-regular graph, we have exactly $m=\frac{(D+1)n}{2}$ edges.
We can then use \Cref{tab:classical_algos_regular_graphs} to compute the energy, once we specify the maximum matching size $M$, and the maximum cut size $C$.

We handle the $2$-regular case separately from the other cases. As discussed in \cref{sec:iterations}, we define random regular graphs to be chosen uniformly over all labeled regular graphs. Under this definition, the graphs are not required to be connected. A random graph of degree $D + 1 = 2$ can be seen as a random permutation of $n$ elements. With high probability, there are at most $O(\log n)$ disjoint cycles, and so we can whp construct a matching of size at least $\frac{n}{2} - O(\log n)$. Similarly, we can whp construct a cut of size at least $\frac{n}{2} - O(\log n)$.

Now we handle random $(D+1)$-regular graphs when $(D+1) \ge 3$. With high probability over the graph, a matching of size $\frac{n}{2}-O(\log n)$ can be found in time polynomial in $n$~\cite{anastos2021}. 
When $D \rightarrow \infty$, we can efficiently find the maximum cut whp over the graph~\cite{alaoui2023}. This has expected size $C=\frac{n(D+1)}{4}+n\frac{P_\star \sqrt{D+1}}{2}$~\cite{alaoui2023}, where $P_\star \approx 0.7631...$ is derived via the Parisi formula from spin glass theory. 
At finite degree, we use known upper bounds on the expected maximum cut size from~\cite{coja-oghlan2020}, and lower-bound C with an explicit algorithm finding good cuts of high-girth regular graphs~\cite{thompson2022explicit}.

\subsection{Small degree}\label{sec:results/small_d}

We consider the QMC, XY, and EPR Hamiltonians on $(D+1)$-regular graphs for each $D\in\ofc{1,2,3,4}$.
In every situation, we optimize the parameters on the MC ansatz for $p=\ofc{1,2,3,4,5}$ and on the XY ansatz for $p=\ofc{1,2}$. We compare them to the values of classical approaches $ZERO$, $MATCH$, and $CUT$ in \Cref{tab:small_d}. Optimal parameters for the ans\"atze are given in \cite{sud2024}.

\begin{table}[t]
\centering
\resizebox{\columnwidth}{!}{
\setlength\tabcolsep{.5mm}
\renewcommand{\arraystretch}{1.4} 
\begin{tabular}{|cl|ccccccc|ccc|}
\hline
& & $MC_{(1,H)}$ & $MC_{(2,H)}$ & $MC_{(3,H)}$ & $MC_{(4,H)}$ & $MC_{(5,H)}$ & $XY_{(1,H)}$ & $XY_{(2,H)}$ &\text{ZERO} & \text{CUT} & \text{MATCH}\\
\hline
\multirow{3}{*}{$D=1$} 
  & $H=\frac{\text{QMC}}{m}$ & 0.5       & 0.8119     & 0.9117     & 1.0134     & 1.0554     & 0.5       & 0.9941     & 0 & 1 & \cellcolor{lgreen}1.25 \\
  & $H=\frac{\text{XY}}{m}$  & 0.5       & 0.8248     & 0.8536     & 0.9408     & 0.9553     & 0.5       & 0.9239     & 0 & \cellcolor{lgreen}1 & \cellcolor{lgreen}1 \\
  & $H=\frac{\text{EPR}}{m}$ & 1.3090    & 1.3493     & 1.3594     & 1.3626     & \cellcolor{lgreen}1.3638     & 1.3493     & 1.3626     & 1 & 1 & 1.25 \\
\hline
\multirow{3}{*}{$D=2$} 
  & $H=\frac{\text{QMC}}{m}$ & 0.5       & 0.6849     & 0.7478     & 0.7938     & 0.8431     & 0.5       & 0.7604     & 0 & [0.8918-0.9241] & \cellcolor{lgreen}1 \\
  & $H=\frac{\text{XY}}{m}$  & 0.5       & 0.7041      & 0.7677      & 0.8032     & 0.8346     & 0.5       & 0.7714     & 0 & \cellcolor{lgreen}[0.8918-0.9241] & 0.8333 \\
  & $H=\frac{\text{EPR}}{m}$ & 1.1922    & 1.2105     & 1.2132     & 1.2140      &\cellcolor{lgreen} 1.2144     & 1.2105     & 1.2140      & 1 & [0.8918-0.9241] & 1 \\
\hline
\multirow{3}{*}{$D=3$} 
  & $H=\frac{\text{QMC}}{m}$ & 0.5       & 0.6622     & 0.7144     & 0.7474     & 0.77425    & 0.5       & 0.72125    & 0 & [0.8333-0.8683] & \cellcolor{lgreen} 0.875 \\
  & $H=\frac{\text{XY}}{m}$  & 0.5       & 0.6690      & 0.7220      & 0.7530    & 0.7730    & 0.5       & 0.7241    & 0 & \cellcolor{lgreen}[0.8333-0.8683] & 0.75 \\
  & $H=\frac{\text{EPR}}{m}$ & 1.1395   & 1.1496    & 1.1507    & 1.1510    & \cellcolor{lgreen}1.1512    & 1.1496    & 1.1510      & 1 & [0.8333-0.8683] & 0.8750\\
\hline
\multirow{3}{*}{$D=4$} 
  & $H=\frac{\text{QMC}}{m}$ & 0.5       & 0.6431    & 0.6906    & 0.7197    & 0.7448    & 0.5       & 0.6926    & 0 & \cellcolor{lgreen}[0.7951-0.8350] & \cellcolor{lgreen} 0.8 \\
  & $H=\frac{\text{XY}}{m}$  & 0.5       & 0.6479    & 0.6944    & 0.7217    & 0.7368     & 0.5       & 0.6961    & 0 & \cellcolor{lgreen}[0.7951-0.8350] & 0.7 \\
  & $H=\frac{\text{EPR}}{m}$ & 1.1094   & 1.1156    & 1.1162    & 1.11636    & \cellcolor{lgreen}1.11639     & 1.1156    &  1.11635    & 1 & [0.7951-0.8350] & 0.8 \\
\hline
\end{tabular}
}
\caption{\footnotesize Per-edge energy obtained by the MC ansatz up to depth $5$, the XY ansatz up to depth $2$, and the classical algorithms ZERO, CUT, and MATCH on regular graphs of degree $D+1$ with $D\in [1,4]$. We report $4$ digits after the decimal point, except in ties. As described in \cref{sec:iterations}, the notation $\text{A}_{(p, \text{H})}$ refers to the energy obtained by ansatz A with depth $p$ on Hamiltonian H. Cells highlighted in green denote to the best algorithm for $H$ at each value of $D$. The upper bounds for CUT are from \cite{coja-oghlan2020}; the lower bounds are from \cite{thompson2022explicit}. 
} 
\label{tab:small_d}
\end{table}

From \Cref{tab:small_d} we can see that for the QMC and XY Hamiltonian, the MC and XY ans\"atze do not outperform MATCH or CUT. For QMC we observe that when $D<4$, MATCH outperforms even the existential upper bound for CUT from \cite{coja-oghlan2020} (which may not be achievable in polynomial time); when $D=4$, MATCH lies between this upper bound on CUT and the algorithmic lower bound from \cite{thompson2022explicit}. (When $D>4$, the algorithmic lower bound on CUT is higher than MATCH, and so CUT outperforms MATCH.) For the XY Hamiltonian, the algorithmic lower bound of CUT always outperforms MATCH, except at $D=1$, where the algorithms yield equal energies. For the EPR Hamiltonian, MATCH outperforms ZERO at $D=1$, the algorithms are equal at $D=2$, and ZERO outperforms MATCH at $D>2$. However, for this Hamiltonian, \emph{the MC ansatz at five layers outperforms all other algorithms for all $D$}. 

The values in \Cref{tab:small_d} also apply to high-girth bipartite regular graphs, as iterations in \cref{sec:iterations} are not altered by the requirement that the graphs have no odd cycles. For bipartite graphs, \Cref{rem:bipartite_qmc_epr_equivalence} tells us that QMC and EPR are equivalent, up to local rotations. Thus, any algorithm for EPR yields an algorithm for QMC obtaining the same energy. For QMC on bipartite graphs then, we compare the values of MATCH and ZERO on EPR with that of MC\textsubscript{(p, EPR)} in \cref{fig:optimal_algs}. Thus, we show that we can outperform classical algorithms at any finite degree.

As described in \cref{sec:background/classical_algos}, we also consider the algorithm KING, an SDP-based approach, which is to our knowledge the only other proposed algorithm for maximizing the EPR Hamiltonian \cite{king2023}. KING outputs the greater of the algorithm ZERO or a one layer quantum circuit composed of two-body Pauli rotations on each edge $(ij)$ parameterized by some $\theta_{ij}$, where the parameters $\theta_{ij}$ are derived from the solution of the SDP. This solution depends on global properties of the graph, so  it is unclear how to analyze this approach on arbitrary high-girth regular graphs. 
However, for regular trees, the edge-transitivity forces parameters $\theta_{ij}=\theta$ to be identical for all $(ij)\in E$. We can then optimize one variable $\theta$ in \cite[Lemma 12]{king2023} to find the energy obtained by KING. We find that for the EPR Hamiltonian on regular trees of any degree, the optimal $\theta$ returns an energy equal to that the MC ansatz at depth $1$ for all $D$ we consider in \cref{tab:small_d}. As the energies obtained by the ansatz are strictly larger at depth two, we outperform the algorithm for for any $p\geq 2$.

\subsection{Infinite degree limit}\label{sec:results/infinite_d}

For the EPR Hamiltonian, 
 our algorithms are matched by simple classical approaches in the $D \rightarrow \infty$ limit, as discussed in \cref{sec:iterations/mc_ansatz/d_inf} and \cref{sec:iterations/xy_ansatz/d_inf}.
For the QMC Hamiltonian, item 2) of \cref{cor:mc_inf_d_energy,cor:xy_inf_d_energy} show that the normalized energy for each algorithm is upper-bounded by that on the XY Hamiltonian.
Thus, we investigate the normalized energy achieved by our algorithms in the infinite degree limit \emph{only} for the XY Hamiltonian.
As we will see, at the depths we can probe numerically, our quantum algorithms are no better than simply finding a good cut (CUT). 

For the MC ansatz, we numerically obtain the values from \cref{eq:mc_nu_infinite_d} in \Cref{tab:mc_nu_values}. We visualize the optimal parameters in \cref{fig:angles_mc} and list all optimal parameters in \Cref{apx:angles}. 
\begin{figure}[h]
    \centering
    \includegraphics[width=0.85\linewidth]{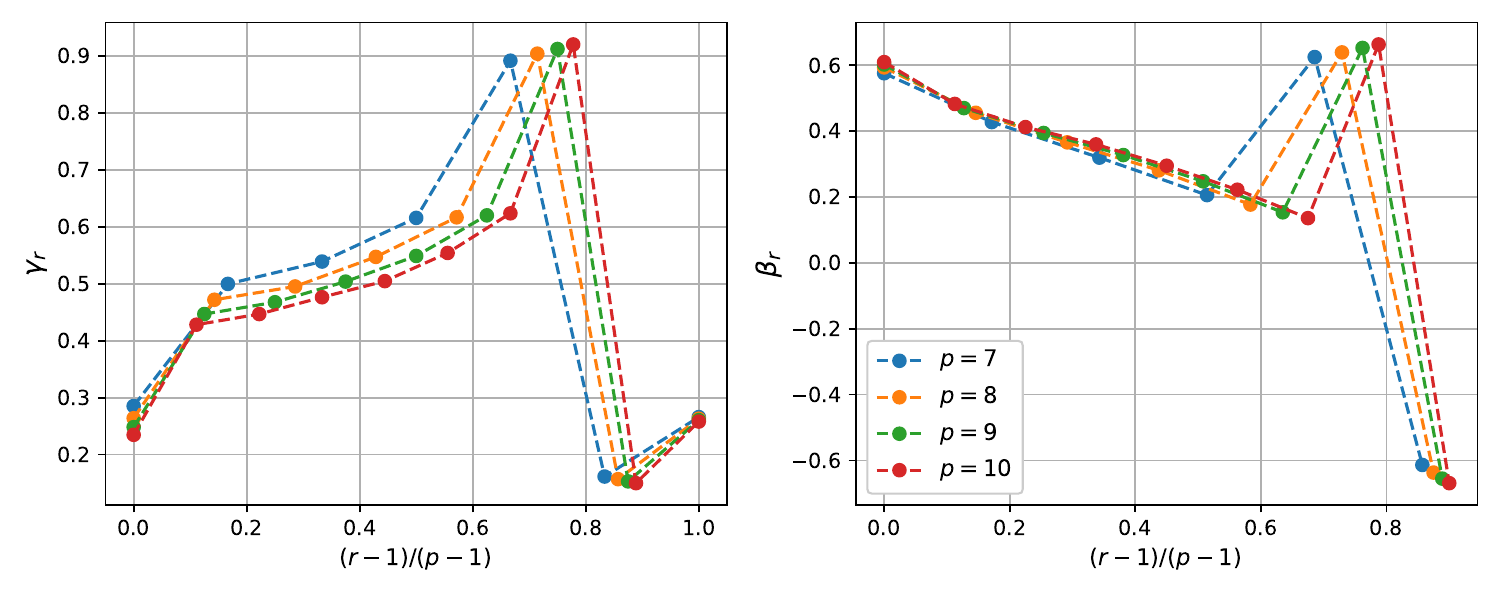}
    \caption{\footnotesize Optimal angles by ansatz layer $r$ for the MC ansatz on the XY Hamiltonian in the infinite degree limit at various depths $p$.}
    \label{fig:angles_mc}
\end{figure}

\begin{table}[ht]
    \centering
\noindent
\resizebox{.7\textwidth}{!}{%
\begin{tabular}{|c | c c c c c |}
    \hline
    p & 1 & 2 & 3 & 4 & 5  \\
    \hline
    $\nu_p(XY, MC)$   & 0 & 0.3086 & 0.4099 & 0.4726 & 0.5157 \\
    $\nu_p(MC, MC)$ (\cite{basso2022}) & 0.3033 & 0.4075 & 0.4726 & 0.5157 & 0.5476 \\
    \hline
    \end{tabular}
    }
    \vspace{1em} % Adds vertical space between the tables
    \noindent
    \resizebox{.7\textwidth}{!}{%
    \begin{tabular}{|c | c c c c c |}
        \hline
        p & 6 & 7 & 8 & 9 & 10 \\
        \hline
        $\nu_p(XY, MC)$   & 0.5476 & 0.5799 & 0.6093 & 0.6321 & 0.6503 \\
        $\nu_p(MC, MC)$ (\cite{basso2022}) & 0.5721 & 0.5915 & 0.5915 & 0.6073 & 0.6203 \\
        \hline
    \end{tabular}
    }
    \caption{\footnotesize Optimal values of $\nu_p$ for the MC ansatz applied to the XY Hamiltonian (top)  up to depth $p=10$.
    We contrast these values to the ansatz's performance on the MC Hamiltonian (bottom), calculated in \cite{basso2022}. Note that these values are quite similar; we do not fully understand why this occurs.}
    \label{tab:mc_nu_values}
\end{table}

Inspecting \Cref{tab:mc_nu_values}, we see that at $p \le 7$, the depth-$p$ performance on the XY Hamiltonian is similar to the depth-$(p-1)$ performance on the (classical) MC Hamiltonian. However, this pattern changes at higher depth; here, we see that the energy achieved on the XY Hamiltonian is greater than that achieved on the MC Hamiltonian.
It is possible that the ansatz must have enough layers before it can do something that outperforms a good product state solution.

At $p > 7$, we also observe that the first several optimal parameters for the XY Hamiltonian nearly match the optimal parameters for the MC Hamiltonian. We do not understand the reason for this behavior, and suggest it as a topic for future research. As an example, see \Cref{fig:p_10_mc_angles_on_H_mc_vs_H_xy}.

\begin{figure}[h]
    \centering
    \includegraphics[width=0.85\linewidth]{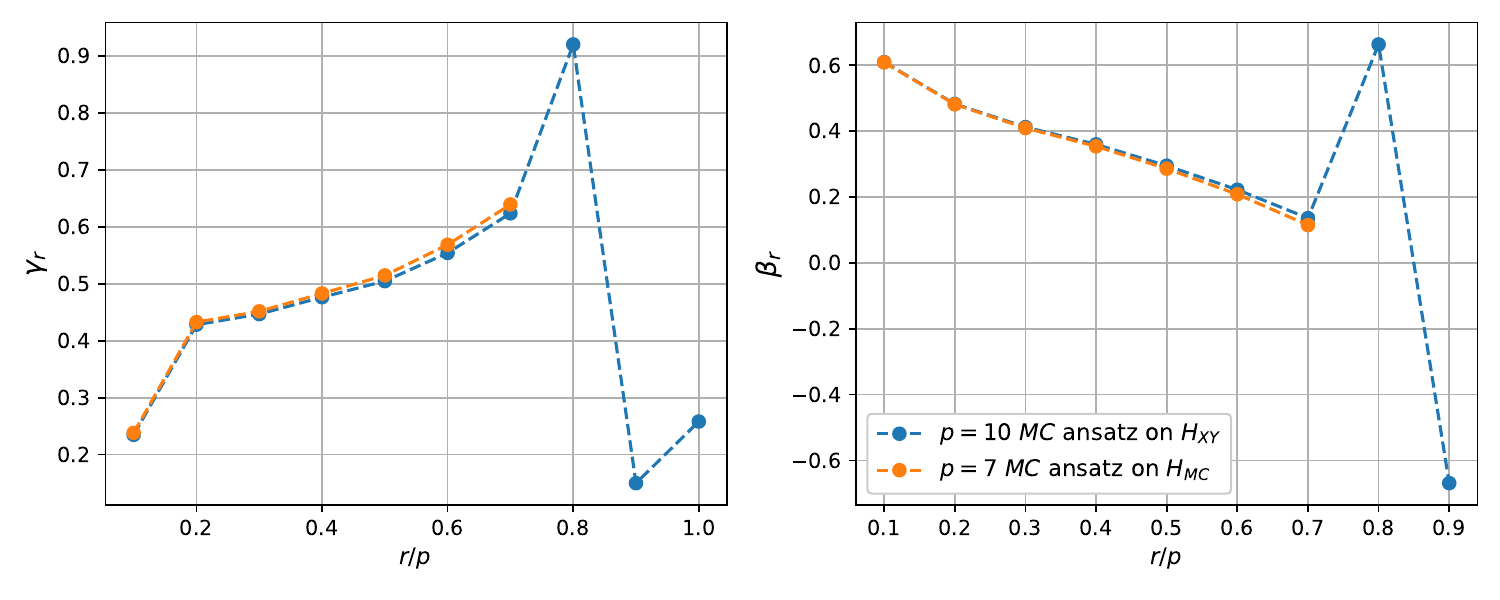}
    \caption{\footnotesize We compare the optimal angles $\gammav,\betav$ by ansatz layer $r$ in the infinite degree limit for two algorithms: (1) $p=10$ MC ansatz on the $H_{XY}$ Hamiltonian, and (2) $p=7$ MC ansatz on the $H_{MC}$ Hamiltonian (from \cite{basso2022}). Note that the first few optimal parameters are almost exactly equal.}
    \label{fig:p_10_mc_angles_on_H_mc_vs_H_xy}
\end{figure}

For the XY ansatz, we numerically obtain the values for \cref{eq:xy_nu_infinite_d} in \Cref{tab:xy_nu_values}. We visualize the optimal parameters in \cref{fig:angles_xy}, and list optimal parameters in \Cref{apx:angles}.

\begin{figure}[ht]
    \centering
    \includegraphics[width=\linewidth]{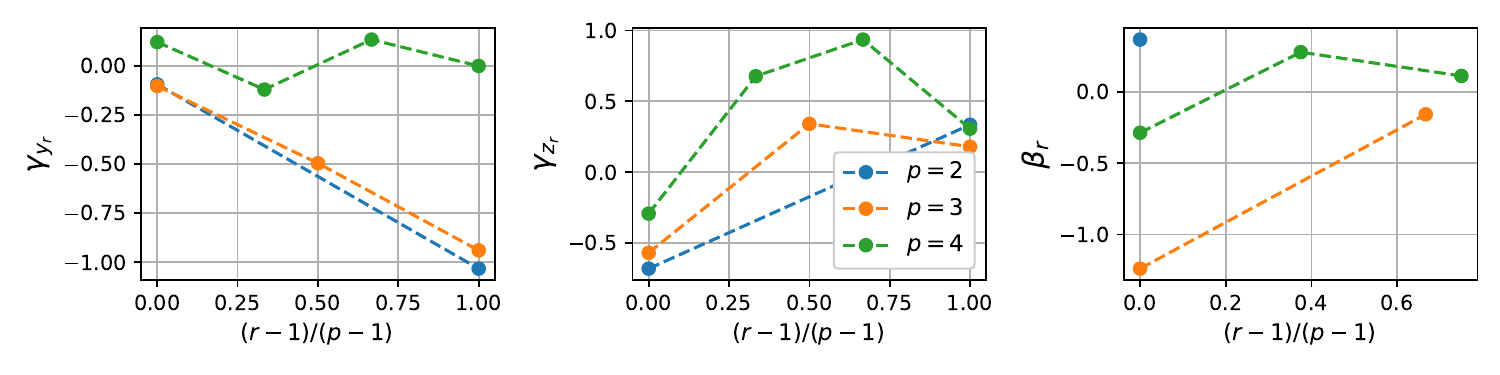}
    \caption{\footnotesize Optimal angles $\gammav_y$, $\gammav_z$ and $\betav_z$ by ansatz layer $r$ for the XY ansatz on the XY Hamiltonian in the infinite degree limit for various depths $p$. Here, we no longer recover parameter schedules with monotonically increasing $\gamma$ and decreasing $\beta$ inspired by quantum annealing, and thus no longer qualitatively match \cite{basso2022}.}
    \label{fig:angles_xy}
\end{figure}
\begin{table}[h!]
    \centering
    \begin{tabular}{|c | c c c c |}
    \hline
        p & 1 & 2 & 3 & 4  \\
        \hline
        $\nu_p$ for $H_{XY}$  & 0. & 0.40611131 & 0.48085015 & 0.515497  \\
    \hline
    \end{tabular}
    \caption{\footnotesize Optimal values for $\nu_p$ with the XY ansatz on the XY Hamiltonian up to depth $p=4$.}
    \label{tab:xy_nu_values}
\end{table}

We summarize the performance of both ans\"atze in \cref{fig:nus}. We compare each value $\nu_p$ to the normalized energy of CUT on high-girth regular graphs, which approaches the Parisi value $P_\star\approx0.7631...$ in the infinite degree limit~\cite{alaoui2023}. 
In general, we find that neither the MC or XY ans\"atze outperform CUT at the depths studied. However, it is unclear what may happen at larger depth, as the runtime of our analysis scales exponentially with depth.
\begin{figure}[t]
    \centering
    \includegraphics[width=0.9\linewidth]{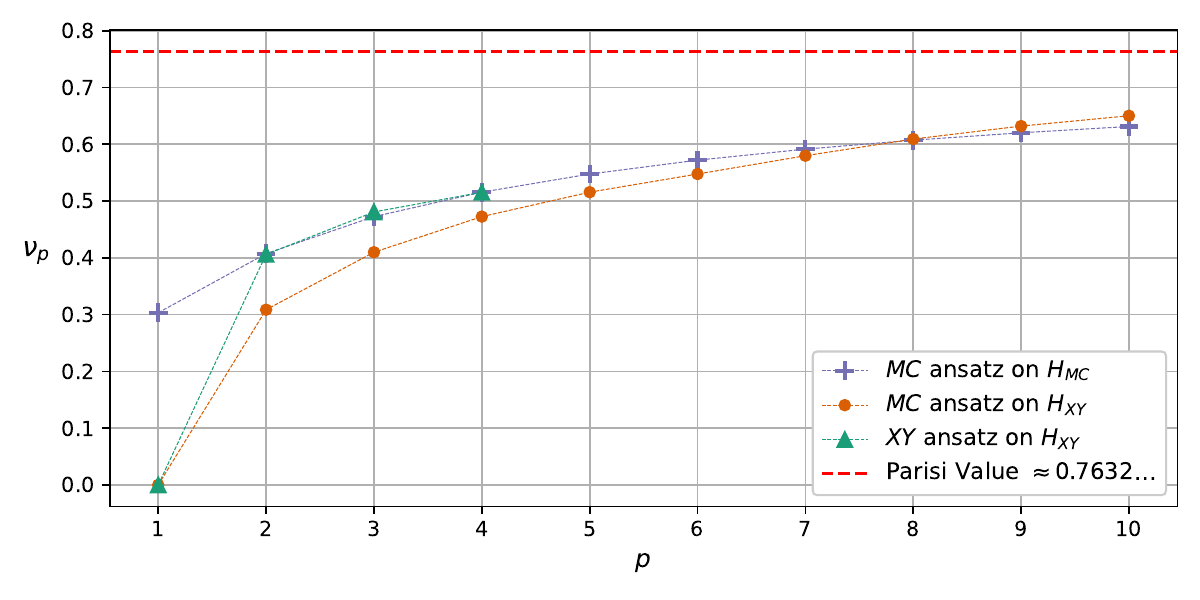}
    \caption{\footnotesize Optimal values $\nu_p$ obtained by the MC and XY ansatz on the XY Hamiltonian. The optimal values $\nu_p$ obtained by MC on the MC Hamiltonian from \cite{basso2022} are given for comparison, as well as a dashed horizontal line indicating the Parisi value, which corresponds to the value of CUT obtained by the algorithm of \cite{alaoui2023}. }
    \label{fig:nus}
\end{figure}

\section{Discussion}\label{sec:discussion}
Despite much effort, there has been limited evidence for quantum advantage on interesting combinatorial optimization problems~\cite{barak2015beating,chou2021limitations, chen2023local,montanaro2024quantum}.
A natural direction, then, is to  look for quantum advantage on \emph{non-commuting} Hamiltonians, which are fundamentally quantum in nature.
We hope this work furthers a competition between classical and quantum algorithms on these local Hamiltonian problems, both for deciding the value of the maximum energy, \emph{and} in preparing states close to this energy value.

In this work, we design and analyze the MC and XY ans\"atze, two variational quantum algorithms to optimize the Quantum MaxCut, XY, and EPR Hamiltonians, which are three specific classes of $2$-Local Hamiltonians. We study the performance of these algorithms for unweighted high-girth regular graphs. Our results indicate that the quantum algorithms, even with numerically-optimized parameters and modest depths (between $4$ and $10$), do not outperform simple classical approaches (based on trivial product states, maximum cuts, and maximum matchings) when maximizing the QMC and XY Hamiltonian at any graph degree. However, we find evidence that the MC ansatz outperforms these simple classical approaches, as well as the semi-definite programming algorithm from King \cite{king2023}, when maximizing the EPR Hamiltonian at small degree. 

We believe that this discrepancy between EPR and QMC/XY may be partly explained in terms of initial states. For the EPR Hamiltonian, the optimal product state is any permutation-invariant product state (for example, the all-zero state $\ket{0}^{\otimes n}$). For the QMC and XY Hamiltonians, the optimal product state may be highly nontrivial \cite{hwang2022}, while permutation-invariant states give us energy $0$. Our framework for analysis in \cref{sec:iterations} requires our variational algorithms to be initialized with a permutation-invariant product state. Thus, for QMC and XY, we start off with the \emph{worst possible} product state, while for EPR we start off with the \emph{best possible} product state. Our results for EPR may be interpreted as: given a good product state, the MC ansatz is able to distribute entanglement better than a SDP and short circuit (\cite{king2023}), a maximum matching (MATCH), or doing nothing at all (ZERO). For QMC and EPR, it may be the case that starting from a good product state, similar statements may hold.
For variational algorithms and the QAOA for classical optimization problems, this is sometimes referred to as ``warm-starting''~\cite{augustino2024, tate2024}. Indeed, concurrent work~\cite{kannan2024} warm-starts their variational ans\"atze with product states obtained from good classical cuts, for example via \cite{goemans1995}; they show this circuit may further improve the energy.

In the infinite degree limit, our algorithms fail to outperform classical methods. Intuitively, this is because of the \emph{monogamy of entanglement} phenomenon. In \cite[Lemma 5]{anshu2020}, the authors show that the optimal energy of QMC on an unweighted star graph with a degree-$D$ central vertex is $D+1$. The value of CUT on this graph is $D$. Thus, in the large $D$ limit, CUT --- which yields a product state --- is asymptotically optimal. The same holds for the EPR Hamiltonian (by \cref{rem:bipartite_qmc_epr_equivalence}). One interpretation is that the Hamiltonians aim to project all pairs of vertices into the \emph{maximally entangled} Bell state, but the central vertex cannot be simultaneously entangled with all of its neighbors. For regular graphs, as $D$ grows large, the central vertex is overwhelmed, and so product states become near-optimal. This phenomenon is studied in more detail in \cite{brandao2014}.

Our results for infinite 2-regular graphs correspond to finding low-energy states of the antiferromagnetic Heisenberg XXX\textsubscript{1/2} model on an infinite ring. These models are well studied; for instance, we know from \cite{bethe1931, faddeev1996} that the maximum normalized energy (and energy density, since $m = n$) for QMC in the infinite-size limit is $2\ln2 \approx 1.3863$. These graphs are bipartite (or nearly bipartite for odd $n$), and thus our numerics for the MC ansatz on the EPR Hamiltonian show that we can achieve energy $1.3638$ with just five layers of the algorithm. This is within $1.62\%$ error of the true ground state. By contrast, preparing an \emph{exact} ground state seems challenging: rigorous approaches to prepare Bethe ansatz states (eigenstates of the Hamiltonian \cite{bethe1931, faddeev1996}) on a quantum computer seem to scale unfavorably with system size \cite{dyke2021, li2022,sopena2022a, ruiz2024, ruiz2024b, raveh2024}.
It is surprising that our algorithm can prepare \emph{near}-ground states while the Bethe ansatz methods fail; this discrepancy warrants further study. 

\section*{Acknowledgements}
The authors would like to thank Joao Basso, Sami Boulebnane, Ishaan Kannan, Robbie King, and Leo Zhou for helpful discussions.
This work was done in part while K.M. was visiting the Simons Institute for the Theory of Computing, supported by NSF QLCI Grant No. 2016245.
K.M acknowledges support from AFOSR
(FA9550-21-1-0008).
This material is based upon work supported by the National Science Foundation Graduate Research Fellowship under Grant No. 2140001.
Any opinion, findings, and conclusions or recommendations expressed in this material are those of the authors(s) and do not necessarily reflect the views of the National Science Foundation.
A.S. acknowledges support of the National Science and Engineering Research Council of Canada (NSERC), Canada Graduate Scholarship.

\printbibliography

\appendix

\clearpage
\newpage

\section{Energy of MC ansatz with one layer on arbitrary graphs}
\label{apx:mcansatz_onelayer}
For MaxCut, there are numerous results about the expected performance of the QAOA algorithm. For instance \cite{wang2018} proved several formula for the expected performance of the $MC$ ansatz on MaxCut at depth $p = 1$, which we cast into our language below:

\begin{theorem}[\cite{wang2018} Energy of $p=1$ QAOA on MaxCut]\label{thm:mc_mc_depth_1}
Let $\ket{\gamma, \beta}$ be the state output $p = 1$ QAOA for max-cut, then we have 

\begin{enumerate}[label=(\alph*)]
    \item The expected energy of edge $(u, v)$ on state $\ket{\gamma, \beta}$  
    \begin{align}
        \frac{1}{2} + \frac{1}{4} (\sin 4\beta \sin \gamma) (\cos^{d_u} \gamma + \cos^{d_v} \gamma) - \frac{1}{4} (\sin^2 2\beta \cos^{d_u + d_v - 2 t_{uv}} \gamma) (1 - \cos^{t_{uv}} 2 \gamma),
    \end{align}
    where $d_u + 1$ is the degree of $u$, $d_v + 1$ is the degree of $v$, and $t_{uv}$ is the number of triangles in $G$ including vertices $(u, v)$. 
    \item 
    For a triangle-free, $d + 1$-regular graph $G$, the expected value for $p = 1$ QAOA is
    \begin{align}
        \langle  \psi(\gamma, \beta) \vert H_G^{MC} \vert \psi(\gamma, \beta) \rangle =  \frac{|E|}{2} \left (1 + \sin 4 \beta \sin \gamma \cos^d \gamma \right).
    \end{align}
    \item The optimal values of angles $\beta, \gamma$ in (b) satisfy $\beta = \frac{\pi}{8}, \gamma = \tan^{-1} (\frac{1}{\sqrt{d}}),$ and at these points the expression in (b) achieves value at least 
    \begin{align}
        \frac{|E|}{2} \left (1 + \frac{1}{\sqrt{e (d+1)}} \right).
    \end{align}
\end{enumerate}

\end{theorem}

We now extend these results to the non-commuting Hamiltonians considered in this work. First, let $O$ be an observable in the set $O \in \cup_{(u,v)\in E}\ofc{X_uX_v, Y_uY_v, Z_uZ_v}$. By definition, we can compute MC\textsubscript{H} for $\text{H} \in \ofc{\text{H}_{\text{QMC}}, \text{H}_{\text{XY}}, \text{H}_{\text{EPR}}}$ on any graph given that we can compute $\braket{\gamma, \beta | O | \gamma, \beta}$ for all $O$. Using similar formulas and notation as \Cref{thm:O_mc_depth_1}, we provide a closed for for these expectation values using the MC ansatz below.

\begin{theorem}[Energy of $p=1$ MC ansatz]\label{thm:O_mc_depth_1}
Let $\ket{\gamma, \beta}$ be the state output by the MC ansatz at $p=1$. Then

\begin{enumerate}[label=(\alph*)]
    \item  The expected energy of $\braket{X_uX_v}$, $\braket{Y_uY_v}$, and $\braket{Z_uZ_v}$ for an edge
    $(u, v)$ are given by
    \begin{align}\label{eq:p1_xyz_exp_vals_analytic}
        \braket{\gamma, \beta | O | \gamma, \beta} = \begin{cases}
        \frac{1}{2} \left(\cos ^{t_{uv}}(2 \gamma )+1\right) \cos
        ^{d_u+d_v-2 t_{uv}}(\gamma ), \quad O={X_uX_v}, \\
        -\frac{1}{2}\Big(\sin{4\beta}\of{\cos^{d_u}\gamma+\cos^{d_v}\gamma}\sin{\gamma}\\
        \hspace{20px}-\cos^{2}2\beta\cos^{d_u+d_v-2t_{uv}}\of{1-\cos^{t_{uv}}2\gamma}\Big), \quad O={Y_uY_v}, \\
        \frac{1}{2}\big(\sin{4\beta}\of{\cos^{d_u}\gamma+\cos^{d_v}\gamma}\sin{\gamma}\\
        \hspace{20px}-\sin^{2}2\beta\cos^{d_u+d_v-2t_{uv}}\of{1-\cos^{t_{uv}}2\gamma}\Big), \quad O={Z_uZ_v}, \\
        \end{cases}
    \end{align}
    where $d_u + 1$ is the degree of $u$, $d_v + 1$ is the degree of $v$, and $t_{uv}$ is the number of triangles in $G$ including vertices $(u, v)$. 
    \item For a triangle-free $d + 1$-regular graph $G$, \cref{eq:p1_xyz_exp_vals_analytic} can be simplified to
    \begin{align}\label{eq:p1_reg_triangle_free_xyz_exp_vals_analytic}
        \braket{\gamma, \beta | O | \gamma, \beta} = \begin{cases}
        \cos^{2d}\gamma, \quad O={X_uX_v}, \\
        -\sin4\beta\cos^d\gamma \sin\gamma, \quad O={Y_uY_v}, \\
        \sin4\beta\cos^d\gamma \sin\gamma, \quad O={Z_uZ_v}. \\
        \end{cases}
    \end{align}
\end{enumerate}
\end{theorem}

\begin{proof}
We first prove part (a). note that

\begin{align}\label{eq:cyclic_trace}
Tr(O U_B U_C \rho U_C^* U_B^*) = Tr(U_B^* O U_B U_C \rho U_C^*),
\end{align}

for each $O$ where $\rho = \ket{+}^{\otimes n}  \bra{+}^{\otimes n}$. We then make the claim
\begin{claim}
\label{claim:trace_mix}
    \begin{align}
        U_B^* O U_B = 
        \begin{cases}
            X_uX_v, \quad O=X_uX_v, \\
            (\cos^2 2 \beta) Y_u Y_v - (\cos 2 \beta) (\sin 2 \beta) (Y_u Z_v + Y_v Z_u) + (\sin^2 2\beta) Z_u Z_v, \quad O=Y_uY_v, \\
            (\cos^2 2\beta) Z_u Z_v + (\cos 2\beta) (\sin 2\beta) (Y_u Z_v + Y_v Z_u) + (\sin^2 2\beta) Y_u Y_v, \quad O=Z_uZ_v.
        \end{cases}
    \end{align}
\end{claim}
\begin{proof}
We prove this claim for each $O$ as follows
\begin{itemize}
    \item $Z_uZ_v$: This is given in \cite{wang2018}.
    \item $Y_uY_v$: This is given by an identical argument to the $Z_uZ_v$ case, and using Pauli relations.
    \item $X_uX_v$: This follows naturally as $U_B$ commutes with $X_u X_v$.\qedhere
\end{itemize}
\end{proof}

Combining this claim with \cref{eq:cyclic_trace} we see that to compute \cref{eq:cyclic_trace} we need to evaluate the terms $Tr(Z_uZ_v U_C \rho U_c^*)$, $Tr(Y_uZ_v U_C \rho U_c^*)$, $Tr(Y_uY_v U_C \rho U_c^*)$, and $Tr(X_uX_v U_C \rho U_c^*)$ ($Tr(Z_uY_v U_C \rho U_c^*)$ can then be computed via $Tr(Y_uZ_v U_C \rho U_c^*)$ by symmetry). We analyze these expectations term by term in the following claim, letting $O'\in \cup_{(u,v)\in E} \ofc{Z_uZ_v, Y_uY_v, Z_uY_v, X_uX_v}$.

\begin{claim}\label{claim:trace_cost}
    \begin{align}
        Tr(O' U_C \rho U_c^*) = \begin{cases}
        0, \quad O'=Z_uZ_v, \\
        \frac{1}{2}\cos^{d_u + d_v - 2t_{uv}}\gamma(1-\cos^{t_{uv}}2\gamma), \quad O'=Y_uY_v, \\
        \sin\gamma \cos^{d_v}\gamma, \quad O'=Y_uZ_v, \\       
        \frac{1}{2}\cos^{d_u + d_v - 2t_{uv}}\gamma(1+\cos^{t_{uv}}2\gamma), \quad O'=X_uX_v. \\
        \end{cases}
    \end{align}
\end{claim}

\begin{proof}
    We prove each case in turn:
    \begin{itemize}
        \item $O'=Z_uZ_v$: This case is trivial, as $Z_uZ_v$ commutes with $U_C$, and $Tr(Z_uZ_v \rho)=0$.
        \item $O'=Y_uY_v$: Proof is given in \cite{wang2018}.
        \item $O'=Y_uZ_v$: Proof is given in \cite{wang2018}.
        \item $O'=X_uX_v$: Divide the cost Hamiltonian into $C = C_u + C_v + C'$, where $C_u$ involves all neighbors of $u$, except $v$, and $C_v$ involves all neighbors of $v$ except $u$. By the Pauli commutation relations, observe that
        \begin{align}       
            X_u X_v e^{-i \gamma C_u} e^{-i \gamma C_v} = e^{i \gamma C_u} e^{i \gamma C_v} X_u X_v \quad X_u X_v e^{-i \gamma C'} = e^{-i \gamma C'} X_u X_v.
        \end{align}
        Therefore, since $U_C = e^{-i\gamma C_u} e^{-i \gamma C_v} e^{- i \gamma C'},$ we have
        \begin{align}
            Tr(X_u X_v U_C \rho U_C^*) = Tr(e^{-i \gamma C'} e^{i \gamma C_u} e^{i\gamma C_v} X_i X_j\rho U_C^*) = Tr(e^{2 i \gamma C_u} e^{2 i \gamma C_v} X_i X_j \rho).
        \end{align}
        Then observe
        \begin{align}
            e^{2i \gamma C_u} e^{2i \gamma C_v} = \prod_{i=1}^{d_u} ((\cos \gamma) I + i \sin(\gamma) Z_u Z_{u_i}) \prod_{i=1}^{d_v} ((\cos \gamma) I + i \sin(\gamma) Z_v Z_{v_i}),
        \end{align}
        where $\{u_i\}$ are the $d_u$ neighbors of $u$, except $v$ and $\{v_i\}$ are the $d_v$ neighbors of $v$, except $u$. Let $P$ be a Pauli string appearing in the expansion of $e^{2i \gamma C_u} e^{2i \gamma C_v}$. 
        
        If there exists a vertex $u_i$ chosen with $Z_u Z_{u_i}$ that is a neighbor of $u$, but not $v$, then for that term $Tr(P X_u X_v \rho) = 0$. The same is true if a vertex is chosen that is a neighbour of $v$ but not $u$. Therefore, vertices chosen must be common neighbours of $u$ and $v$, and furthermore, the total vertices chosen must be even, otherwise if there is an odd number of $Z$s weighted on vertex $u$, then $Tr(P X_u X_v \rho) = 0$, since on vertex $u$, we have $Tr(Z_u X_u \lvert + \rangle \langle + \rvert) = 0$. Therefore, we can expand the sum as
        \begin{align}
            Tr(e^{2i \gamma C_u} e^{2i \gamma C_v} X_i X_j \rho) &= \sum_{j = 0, j \; \text{even}}^{t_{uv}} \binom{t_{uv}}{j} (\cos^{d_u + d_v - 2j} \gamma) (i \sin \gamma)^{2j}  \\
            &= (\cos^{d_u + d_v - 2 t_{uv}} \gamma)\sum_{j = 0, j \; \text{even}}^{t_{uv}} \binom{t_{uv}}{i} (\cos^2 \gamma)^{t_{uv} - j} (\sin^2 \gamma)^{j}.
         \end{align}
         Then, we can use a linear combination of two binomial theorems, as in the same way \cite{wang2018}, albeit by adding the two equalities, in order to obtain the form claimed.\qedhere
    \end{itemize}
\end{proof}

The proof of (a) then follows from combining \cref{claim:trace_mix} and \cref{claim:trace_cost} with \cref{eq:cyclic_trace}. Part (b) then follows immediately from part (a) by linearity, and since for every edge $(u, v)$, we have $t_{uv} = 0$ and $d_u = d_v$.
\end{proof}

\clearpage
\newpage

\section{Useful lemmas}
\label{apx:usefullemmas}

In this section, we prove \Cref{lemma:xy_f_of_negative_a_is_a,lemma:yz_sum_fH_is_one}, which are the two nontrivial differences between the analysis in \cref{sec:iterations/xy_ansatz} and that of \cite{basso2022}.

We first note the following relations, using the definitions in \cref{eq:xy_sum_definitions}

\begin{align}
    &\frac{\braket{-a_y|-a_z}}{\braket{a_y|a_z}} =  \frac{\braket{-a_y|e^{i\beta X}|-a_z}}{\braket{a_y|e^{i\beta X}|a_z}} = i^{a_y} \implies \frac{\braket{-a_{z_1}|-a_y}}{\braket{a_{z_1}|a_y}}\frac{\braket{-a_y|e^{i\beta X}|-a_{z_2}}}{\braket{a_y|e^{i\beta X}|a_{z_2}}}=1,\label{eq:paired_unity_relation}\\
    &\frac{\braket{-a_y|a_z}}{\braket{a_y|a_z}} = i^{a_y-a_y a_z}, \quad \frac{\braket{a_y|-a_z}}{\braket{-a_y|a_z}} = i^{-a_y a_z} .\label{eq:y_z_bracket_relations}
\end{align}

\begin{lemma}\label{lemma:xy_f_of_negative_a_is_a}
$f(-\av) = f(\av)$,  where $f(\av)$ is defined in \cref{eq:xy_f_defn}.
\end{lemma}

\begin{proof}
First pair off every term in \cref{eq:xy_f_defn} to get 

\begin{align}
    f(\av) = &\frac{1}{2} \of{\braket{a_{z_1}|a_{y_1}} \braket{a_{y_1}|e^{i \beta_1 X_u}|a_{z_2}}} \cdots \of{\braket{a_{z_p}|a_{y_p}} \braket{a_{y_p}|e^{i \beta_p X_u}|a_{0}}} \\
    &\times \of{\braket{a_{0}|e^{-i \beta_p X_u}|a_{y_{-p}}}\braket{a_{y_{-p}}|a_{z_{-p}}}} \cdots \of{\braket{a_{z_{-2}}|e^{-i \beta_1 X_u}|a_{y_{-1}}}\braket{a_{y_{-1}}|a_{z_{-1}}}},
\end{align} 

and likewise for $f(-\av)$. Then, $\frac{f(-\av)}{f(\av)}$ is equal to the product of the fractions of each pair. Using \cref{eq:paired_unity_relation}, the fraction of each pair is $1$.
\end{proof}

\begin{lemma}\label{lemma:xy_f_prime_of_negative_a_is_a}
$f'(-\av) = f'(\av)$,  where $f(\av)$ is defined in \cref{eq:xy_f_prime_defn}.
\end{lemma}
\begin{proof}
The proof follows the exact same logic as \Cref{lemma:xy_f_of_negative_a_is_a}.
\end{proof}

For the remaining lemmas, we first introduce notation that makes our bitstring indexing for the XY ansatz match exactly with the notation in \cite{basso2022}. First, let 

\begin{align}\label{eq:xy_ansatz_b_labelling}
    B= \{(a_1, a_2,\dots,a_{2p},a_0, a_{-2p},\dots,a_{-2}, a_{-1}): a_j=\pm 1\},
\end{align}

be the set of $(4p+1)$-bit strings, which is a relabelling of the set of length $(4p+1)$ bitstrings used in the XY ansatz. Then we can define $B_0$, $T(\av)$, $\av'$, and $\ddot{\av}$ and  exactly as in \cite{basso2022}, except with length $4p+1$ bitstrings. After using these definitions, we can return to the original labelling under the mapping

\begin{align}\label{eq:xy_labelling}
    r \rightarrow 
    \begin{cases}
        z_{(r+1)/2}, &r>0, \text{odd} \\
        y_{r/2},     &r>0, \text{even} \\
        z_{0},       &r=0  \\
        y_{r/2},     &r<0, \text{even} \\
        z_{(r-1)/2}, &r<0, \text{odd}. \\
    \end{cases}
\end{align}

This allows us to show the equivalent of \cite[Lemma 1]{basso2022}:
\begin{lemma}
For $f(\av)$ as defined in \cref{eq:xy_f_defn}, for any $\av$ such that $\av \neq \ddot{\av}$,
\begin{align}
    f(\av') = -f(\av).
\end{align}
\end{lemma}

\begin{proof}
We show $\frac{f(\av')}{f(\av)} = -1$, by applying \cref{eq:y_z_bracket_relations}, using the labelling in \cref{eq:xy_labelling}. We first show for $T(a)$ even. This means that there is some $r$ such that $\av_{z_s}=\av_{z_{-s}}$ for all $s\geq r+1$ and $\av_{y_s}=\av_{y_{-s}}$ for all $s \geq r$. Furthermore $\av'$ corresponds to flipping all $\av_{y_{s}}$, $\av_{y_{s}}$ for all $s \leq s-1$, and all $\av_{z_{s}}$, $\av_{z_{-s}}$ for all $s \leq r$. Thus we obtain the fraction
\begin{align*}
    &\frac{f(\av')}{f(\av)}=\of{\frac{\braket{-a_{z_1}|-a_{y_1}}\braket{-a_{y_1}|e^{i\beta_1 X}|-a_{z_2}}}{\braket{a_{z_1}|a_{y_1}}\braket{a_{y_1}|e^{i\beta_1 X}|a_{z_2}}}} \cdots  
    \of{\frac{\braket{-a_{z_{r-1}}|-a_{y_{r-1}}}\braket{-a_{y_{r-1}}|e^{i\beta_{r-1} X}|-a_{z_r}}}{\braket{a_{z_{r-1}}|a_{y_{r-1}}}\braket{a_{y_{r-1}}|e^{i\beta_{r-1} X}|a_{z_r}}}} \\
    &\times\of{\frac{\braket{-a_{z_r}|a_{y_r}}}{\braket{a_{z_r}|a_{y_r}}}} \cdot\of{1} \cdots \of{1} \cdot 
     \of{\frac{\braket{a_{y_r}|a_{z_r}}}{\braket{a_{y_r}|-a_{z_r}}}} \\
    &\times \of{\frac{\braket{a_{z_r}|e^{-i\beta_{r-1} X}|-a_{y_{-(r-1)}}}\braket{-a_{y_{-(r-1)}}|-a_{z_{-(r-1)}}}}{\braket{-a_{z_r}|e^{-i\beta_{r-1} X}|a_{y_{-(r-1)}}}\braket{a_{y_{-(r-1)}}|a_{z_{-(r-1)}}}}} \cdots
    \of{\frac{\braket{-a_{z_{-2}}|e^{-i\beta_{-1} X}|-a_{y_{-1}}}\braket{-a_{y_{-1}}|-a_{z_{-1}}}}{\braket{a_{z_{-2}}|e^{-i\beta_{-1} X}|a_{y_{-1}}}\braket{a_{y_{-1}}|a_{z_{-1}}}}},
\end{align*}
where the term $\of{1}\cdots\of{1}$ in the middle line represents the trivial fractions of unity, as in these terms the numerator and denominator are equal by the definition of $\av'$. We then note by \cref{eq:paired_unity_relation} that the only non unity terms are
    \begin{align*}
    \of{\frac{\braket{-a_{z_r}|a_{y_r}}}{\braket{a_{z_r}|a_{y_r}}}} 
    \of{\frac{\braket{a_{y_r}|a_{z_r}}}{\braket{a_{y_r}|-a_{z_r}}}} = i^{a_{y_r} a_{z_r}}i^{a_{y_r} a_{z_r}}=(-1)^{a_{y_r} a_{z_r}}=(-1)^{\pm 1}=-1.
\end{align*}
Then, if $T(a)$ is even, there is some $r$ such that $\av_{y_s}=\av_{y_{-s}}$ and $\av_{z_s}=\av_{z_{-s}}$ for all $s \geq r$. Furthermore $\av'$ corresponds to flipping all $\av_{y_{s}}$, $\av_{y_{s}}$, $\av_{z_{s}}$, $\av_{z_{-s}}$ for all $s \leq r$. Then we can once again pair off terms as we did in the even case, and apply \cref{eq:paired_unity_relation} to obtain the nonzero terms
\begin{align*}
    \frac{f(\av')}{f(\av)}=\of{\frac{\braket{-a_{z_{r}}|-a_{y_{r}}}\braket{-a_{y_{r}}|e^{i\beta_r X}|a_{z_{r+1}}}}{\braket{a_{z_{r}}|a_{y_{r}}}\braket{a_{y_{r}}|e^{i\beta_r X}|a_{z_{r+1}}}}}
    \of{\frac{\braket{a_{z_{r+1}}|e^{-i\beta_r X}|a_{y_{r}}}\braket{a_{y_{r}}|-a_{z_{r}}}}{\braket{a_{z_{r+1}}|e^{-i\beta_r X}|-a_{y_{r}}}\braket{-a_{y_{r}}|a_{z_{r}}}}} \\
    =(-1)^{a_{y_r}(a_{z_r}-1)}(-1)^{a_{y_r}}=(-1)^{a_{y_r}a_{z_r}=(-1)^{\pm1}} = -1.\tag*{\qedhere}
\end{align*}
\end{proof}

We then show the equivalent of \cite[Lemma 2]{basso2022}:
\begin{lemma}
\label{lemma:basso_lemma_2}
 For $f(\av)$ as defined in \cref{eq:xy_f_defn}, for any $\av \in B$:
\begin{align*}
    f(\ddot{\av}) = f(\av)^*
\end{align*}
\end{lemma}

\begin{proof}
This follows by the definition of $f(\av)$ along with noting that, for all $1 \leq r \leq p$: 
\begin{align*}
    \braket{a_{z_r}|a_{y_r}}&\braket{a_{y_r}|e^{i\beta_r X}|a_{z_{r+1}}} \braket{a_{z_{-(r+1)}}|e^{-i\beta_r X}|a_{y_{-r}}}\braket{a_{y_{-r}}|a_{z_{-r}}} \\
    =&\braket{\ddot{a}_{z_{-r}}|\ddot{a}_{y_{-r}}}\braket{\ddot{a}_{y_{-r}}|e^{i\beta_r X}|\ddot{a}_{z_{-(r+1)}}} \braket{\ddot{a}_{z_(r+1)}|e^{-i\beta_r X}|\ddot{a}_{y_{r}}}\braket{\ddot{a}_{y_{r}}|\ddot{a}_{z_{r}}} \\ &=\braket{\ddot{a}_{z_r}|\ddot{a}_{y_r}}^*\braket{\ddot{a}_{y_r}|e^{i\beta_r X}|\ddot{a}_{z_{r+1}}}^* \braket{\ddot{a}_{z_{-(r+1)}}|e^{-i\beta_r X}|\ddot{a}_{y_{-r}}}^*\braket{\ddot{a}_{y_{-r}}|\ddot{a}_{z_{-r}}}^*. \qedhere
\end{align*}
\end{proof}

We then note that Lemmas 3, 4, 5, and 6 from \cite{basso2022} apply in our setting: they follow exactly the exact same logic due to our abstraction of the concepts $f(\av)$, $H(\av)$, $B_0$, $T(\av)$, $\ddot{\av}$, $\av'$, and $\Gammav$. 
We reproduce \cite[Lemma 5]{basso2022} below:
\begin{lemma}\label{lemma:yz_sum_fH_is_one}
For $f(\av)$ as defined in \cref{eq:xy_f_defn} and $H_D^{(m)}(\av)$ as defined in \cref{eq:xy_H_defn}, the following holds for all $m$:
\begin{align*}
    \sum_{\av}  f(\av) H_D^{(m)}(\av) =1.
\end{align*}
\end{lemma}

For one calculation, we need the fact that the analogous sum for $f'$ is real:
\begin{lemma}
\label{lemma:f_prime_h_sum_is_real}
    For $f'(\av)$ as defined in \cref{eq:mc_f_prime_defn} and \cref{eq:xy_f_prime_defn} and $H_D^{(m)}(\av)$ as defined in \cref{eq:mc_H_defn} and \cref{eq:xy_H_defn}, we have $ \sum_{\av}  f'(\av) H_D^{(m)}(\av) \in \mathbb{R}$ for all $m$.
\end{lemma}
\begin{proof}
    We first write 
    \begin{align*}
        \sum_{\av}  f'(\av) H_D^{(m)}(\av)  = \frac{1}{2}\sum_{\av}  \left( f'(\av) H_D^{(m)}(\av)  + f'(\ddot{\av}) H_D^{(m)}(\ddot{\av})  \right)\,.
    \end{align*}
    We complete the proof by showing $ \sum_{\av} f'(\ddot{\av}) H_D^{(m)}(\ddot{\av}) = \left(\sum_{\av}  f'(\av) H_D^{(m)}(\av) \right)^*$. 
    First of all, let $g(\av)$ be the operation that only flips $a_0$. Then we have 
    \begin{align*}
        H_D^{(m)}(\av) = H_D^{(m)}(g(\av))\,, 
    \end{align*}
    by definition of $H_D^{(m)}(\av)$ (since $\Gamma_0 = 0$). Also, by a proof almost identical to \cref{lemma:basso_lemma_2}, we have
    \begin{align*}
        f'(g(\ddot{\av})) = f'(\av)^*\,.
    \end{align*}
    Putting these together, we have 
    \begin{align*}
        \sum_{\av} f'(\ddot{\av}) H_D^{(m)}(\ddot{\av})
        =
        \sum_{\av} f'(g(\ddot{\av})) H_D^{(m)}(g(\ddot{\av}))
        =
        \sum_{\av} f'(\av)^* H_D^{(m)}(\ddot{\av}) =  \sum_{\av} f'(\av)^* H_D^{(m)}(\av)^*\,,
    \end{align*}
    where the last equality follows from $H_D^{(m)}(\ddot{\av}) = H_D^{(m)}(\av)^*$ (our version of \cite[Lemma 4]{basso2022}).
\end{proof}

Finally, we prove one other statement corresponding to \cite[Equation 4.14]{basso2022} in our setting:
\begin{lemma}\label{lemma:yz_H_negative_a_is_H_a}
For $f(\av)$ as defined in \cref{eq:xy_f_defn} and $H_D^{(m)}(\av)$ as defined in \cref{eq:xy_H_defn}, the following holds for all $m$ and all $\av$:
\begin{align*}
    H^{\ofc{m}}(-\av) = H^{\ofc{m}}(\av)
\end{align*}
\end{lemma}
\begin{proof}
We prove by induction. For the base case, clearly $H^{\of{0}}(-\av) = 1 =  H^{\of{0}}(\av)$. Then we assume $H^{\of{m-1}}(-\av) = H^{\of{m-1}}(\av)$. Then we can take $\bv \rightarrow -\bv$ in the summand of \cref{eq:xy_H_defn} and combine it with its original form to find
\begin{align}
    H^{(m)}_D(\av) = \of{\sum_{\bv} f(\bv) H^{(m-1)}_D(\bv) \cos\ofb{\frac{1}{\sqrt{D}}\Gammav\cdot(\av \bv)}}^D.
\end{align}
Now, since $f(-\bv)=f(\bv)$, $H^{(m-1)}_D(-\bv)=H^{(m-1)}_D(\bv)$, and $\cos\ofb{\cdot}$ is even in $\bv$, the induction step holds.
\end{proof}

\clearpage
\newpage

\section{Proofs of iterations}\label{apx:iteration_proofs}
We now prove the statements from \cref{sec:iterations}.
\subsection{MC ansatz}\label{apx:iteration_proofs/mc}
We first prove \cref{thm:mc_ansatz_energy}. We seek the expected energy of the MC ansatz defined in \cref{eq:mc_ansatz} on Hamiltonians of the form in \cref{eq:2local_H_form}. For conciseness, we omit the identity term, as it is trivial. The analysis follows very closely to \cite{basso2022}.
We assume we have a $(D+1)$-regular graph with girth $> 2p+1$, and seek the expected energy of QMC with respect to the MC ansatz at depth $p$. As noted in the introduction, due to the regularity and high-girth of the graph, each edge contributes equally to the energy, so it suffices to find the energy contribution from single edge $(L,R)$
\begin{align} 
    &\braket{\paramv | (c_X X_L X_R + c_Y Y_L Y_R + c_Z Z_L Z_R) | \paramv}  \nonumber\\
    &\hspace{100px}=
    \bra{s} e^{i\gamma_1 C_z} e^{i\beta_1 B} \cdots e^{i \gamma_p C_z} e^{i\beta_p B} \of{c_X X_L X_R+ c_Y Y_L Y_R + c_Z Z_L Z_R} \nonumber\\
    &\hspace{116px} \times e^{-i\beta_p B} e^{-i \gamma_p C_z} \cdots e^{-i\beta_1 B} e^{-i\gamma_1 C_z} \ket{s} \label{eq:mc_ansatz_start}.
\end{align} 

To evaluate \cref{eq:mc_ansatz_start}, we note that in the Heisenberg picture, the operators 
\begin{align*}
    e^{i\gamma_1 C_z} e^{i\beta_1 B} \cdots e^{i \gamma_p C_z} e^{i\beta_p B} X_L X_R e^{-i\beta_p B} e^{-i \gamma_p C_z} \cdots e^{-i\beta_1 B} e^{-i\gamma_1 C_z},
\end{align*}
etc, only depend non-trivially on the subgraph of all vertices at most distance $p$ from $L$ and $R$. Because we
considered a high-girth, $D+1$ regular graph, this subgraph is precisely given by a pair of trees, glued at their roots. (For a visual representation, see \cite[Figure 2]{basso2022}). This subgraph has a total of $n=2\of{D^p+\cdots+D+1}$ vertices, which we label with the set $\ofb{n}$, and we can thus restrict our calculation to these $\ofb{n}$ vertices. We then begin by inserting $2p+2$ identities via complete sets in the computational (Pauli $Z$) basis, each of which iterates over $2^n$ $\pm{1}$-valued bitstrings:
\begin{align}
    &\braket{\paramv| c_X X_L X_R + c_Y Y_L Y_R + c_Z Z_L Z_R|\paramv } \nonumber\\
    &\hspace{10px}=
    \sum_{\{\zv^\ofb{i}\}_{i\in \mathcal{A}}} \braket{s |\zv^\ofb{1}}  
    e^{i\gamma_1 C_z(\zv^\ofb{1})} 
    \braket{\zv^\ofb{1}| e^{i\beta_1 B} |\zv^\ofb{2}} \cdots
    e^{i\gamma_{p} C_z(\zv^\ofb{p})}
    \braket{\zv^\ofb{p}| e^{i\beta_p B} |\zv^\ofb{p+1}}  \nonumber \\
    &\hspace{47px} \times\braket{\zv^\ofb{p+1}| c_X X_L X_R + c_Y Y_L Y_R + c_Z Z_L Z_R |\zv^\ofb{-(p+1)}} \nonumber \\
    &\hspace{47px} \times \braket{\zv^\ofb{-(p+1)}| e^{-i\beta_p B} |\zv^\ofb{-p}}
    e^{-i\gamma_p C_z(\zv^\ofb{-p})} \cdots
    \braket{\zv^\ofb{-2}| e^{-i\beta_1 B} |\zv^\ofb{-1}}
    e^{-i\gamma_1 C_z(\zv^\ofb{-1})} \braket{\zv^\ofb{-1}|s} \nonumber \\
    &\hspace{10px}=
    \sum_{\{\zv^\ofb{i}\}_{i\in \mathcal{A}}} \exp\ofb{i \sum_{j\in{\ofb{p}}} \gamma_j \of{C_z\of{\zv^\ofb{j}}-C\of{\zv^\ofb{-j}}}} \braket{\zv^\ofb{p+1}| c_X X_L X_R + c_Y Y_L Y_R + c_Z Z_L Z_R |\zv^\ofb{-(p+1)}} \nonumber \\
    &\hspace{47px}\times\frac{1}{2^n} \braket{\zv^\ofb{1}| e^{i\beta_1 B} |\zv^\ofb{2}} \cdots \braket{\zv^\ofb{p}| e^{i\beta_p B} |\zv^\ofb{p+1}} \braket{\zv^\ofb{-(p+1)}| e^{-i\beta_p B} |\zv^\ofb{-p}} \cdots \braket{\zv^\ofb{-2}| e^{-i\beta_1 B} |\zv^\ofb{-1}},\label{eq:mc_ansatz_basis} 
\end{align}

where
\begin{align}
    \mathcal{A} &= \ofc{1,2,\ldots,p+1,-(p+1),\ldots,-2,-1}.
\end{align}
We then note the relations
\begin{align} \label{eq:xyz_relations} 
    X\ket{z} &= \ket{-z}, \\
    Y\ket{z} &= iz\ket{-z}, \\
    Z\ket{z} &= z\ket{z},
\end{align} 

for $z \in \{\pm 1\}$,  which allows us to evaluate

\begin{align}
    \braket{\zv^\ofb{p+1}| X_L X_R |\zv^\ofb{-(p+1)}} &=  \delta_{z^{\ofb{p+1}}_L, -z^{\ofb{-(p+1)}}_L} \delta_{z^{\ofb{p+1}}_R, -z^{\ofc{-(p+1)}}_R} \prod_{v \in \ofb{n}/\ofc{L,R}} \delta_{z^{\ofb{p+1}}_v,z^{\ofb{-(p+1)}}_v}, \\
    \braket{\zv^\ofb{p+1}| Y_L Y_R |\zv^\ofb{-(p+1)}} &=  -z^\ofb{-(p+1)}_L z^\ofb{-(p+1)}_R\delta_{z^{\ofb{p+1}}_L, -z^{\ofb{-(p+1)}}_L} \delta_{z^{\ofb{p+1}}_R, -z^{\ofc{-(p+1)}}_R} \prod_{v \in \ofb{n}/\ofc{L,R}} \delta_{z^{\ofb{p+1}}_v,z^{\ofb{-(p+1)}}_v}, \\
    \braket{\zv^\ofb{p+1}| Z_L Z_R |\zv^\ofb{-(p+1)}} &= z^\ofb{-(p+1)}_L z^\ofb{-(p+1)}_R \prod_{v \in \ofb{n}} \delta_{z^{\ofb{p+1}}_v,z^{\ofb{-(p+1)}}_v}.
\end{align}
This can alternatively be written as
\begin{align}
    \braket{\zv^\ofb{p+1}| X_L X_R |\zv^\ofb{-(p+1)}} &= 
    \begin{cases}
    1, & \text{if} \quad {\zv^\ofb{p+1}}^{(L,R)} = \zv^\ofb{-(p+1)} \\
    0, & \text{else}
    \end{cases}, \\
    \braket{\zv^\ofb{p+1}| Y_L Y_R |\zv^\ofb{-(p+1)}} &= 
    \begin{cases}
    -z^\ofb{-(p+1)}_L z^\ofb{-(p+1)}_R, & \text{if} \quad {\zv^\ofb{p+1}}^{(L,R)} = \zv^\ofb{-(p+1)} \\
    0, & \text{else}
    \end{cases}, \\
    \braket{\zv^\ofb{p+1}| X_L X_R |\zv^\ofb{-(p+1)}} &= 
    \begin{cases}
    z^\ofb{-(p+1)}_L z^\ofb{-(p+1)}_R, & \text{if} \quad \zv^\ofb{p+1} = \zv^\ofb{-(p+1)} \\
    0, & \text{else}
    \end{cases} \label{eq:inner_product_nonzero_cases}, 
\end{align}

where ${\zv^\ofb{j}}^{(L,R)}$ denotes the bitstring $\zv^\ofb{j}$ with bits $L$ and $R$ flipped. To simplify the sum, we define $\zv^\ofb{0}\defeq\zv^\ofb{-(p+1)}$, and via the above relations restrict $\zv^\ofb{(p+1)}$ the cases where the summand yields nonzero values, allowing us to remove the sum over the complete set of bitstring $\zv^\ofb{(p+1)}$. This allows us to rewrite \cref{eq:mc_ansatz_basis} as 

\begin{align}
    &\braket{\paramv| c_X X_L X_R + c_Y Y_L Y_R + c_Z Z_L Z_R|\paramv } \nonumber\\
    &\hspace{50px}=
    \sum_{\{\zv^\ofb{i}\}_{i\in \mathcal{A'}}} \exp\ofb{i \sum_{j\in{\mathcal{A}'}} \Gamma_j C\of{\zv^\ofb{j}}} \frac{1}{2^n} \prod_{j=1}^{p-1} \braket{\zv^\ofb{j}| e^{i\beta_j B} |\zv^\ofb{j+1}} \braket{\zv^\ofb{-(j+1)}| e^{-i\beta_j B} |\zv^\ofb{-j}} \nonumber\\
    &\hspace{90px}\times \Bigg( c_X \braket{\zv^\ofb{p}| e^{i\beta_p B} |{\zv^\ofb{0}}^{(L,R)}} \braket{\zv^\ofb{0}| e^{-i\beta_p B} |\zv^\ofb{-p}}\nonumber \\
    &\hspace{110px}- c_Y z^\ofb{0}_L z^\ofb{0}_R\braket{\zv^\ofb{p}| e^{i\beta_p B} |{\zv^\ofb{0}}^{(L,R)}} \braket{\zv^\ofb{0}| e^{-i\beta_p B} |\zv^\ofb{-p}}\nonumber \\
    &\hspace{110px}+ c_Z z^\ofb{0}_L z^\ofb{0}_R\braket{\zv^\ofb{p}| e^{i\beta_p B} |{\zv^\ofb{0}}} \braket{\zv^\ofb{0}| e^{-i\beta_p B} |\zv^\ofb{-p}}\Bigg)\label{eq:mc_ansatz_before_bit_indexing}, 
\end{align}

where the third, fourth, and fifth line represent the contribution from $\braket{X_LX_R}$, $\braket{Y_LY_R}$, and $\braket{Z_LZ_R}$, respectively, and we defined
\begin{align}
    \mathcal{A}' &= \ofc{1,2,\ldots,0,\ldots,-2,-1},\\
    \Gammav &= (\gamma_1, \gamma_2, \ldots, \gamma_p, 0, -\gamma_p, \ldots, -\gamma_2, -\gamma_1).
\end{align}
In this formulation, we have a sum over $2p+1$ bitstrings of length $n$ each. We may instead replace this sum with an equivalent sum over the over $n$ length-($2p+1$) complete basis sets (one for each node $u$ in the subgraph). This substitution simply indexes over bit assignments $z^\ofb{i}_u$ vertex-by-vertex, rather than layer-by-layer. We then apply from \cref{sec:background/quantum_algos} that $C_z(z) \defeq -\frac{1}{\sqrt{D}}\sum_{(u,v)\in E}z_uz_v$ to rewrite the sum as

\begin{align}
    &\braket{\paramv| c_X X_L X_R + c_Y Y_L Y_R + c_Z Z_L Z_R|\paramv } \nonumber\\
    &\hspace{50px}=
    \sum_{{\ofc{\zv_u}}_{u\in \ofb{n}}} \exp\ofb{-\frac{i}{\sqrt{D}}  \sum_{(u',v')\in E} \Gammav  \cdot (\zv_{u'} \zv_{v'})} \of{\prod_{v \in \ofb{n} / \ofc{L,R}} f\of{\zv_v}} \nonumber \\
    &\hspace{80px} \times \Big(c_X f'(\zv_L)f'(\zv_R)-c_Y \zv_L\zv_Rf'(\zv_L)f'(\zv_R)+c_Z \zv_L\zv_Rf(\zv_L)f(\zv_R)\Big)\label{eq:mc_ansatz_bitstring_sum},
\end{align}

where 
\begin{align} 
    f(\av) &\defeq \frac{1}{2} \braket{a_1 | e^{i \beta_1 X} | a_2} \cdots \braket{ a_{p-1} | e^{i \beta_{p-1} X} | a_p} \braket{a_p | e^{i \beta_p X} | a_0} \nonumber \\
    &\quad \times \braket{a_{0} | e^{-i \beta_p X} | a_{-p}} \braket{a_{-p} | e^{-i \beta_{p-1} X} | a_{-(p-1)}} \cdots \braket{a_{-2} | e^{-i \beta_1 X} | a_{-1}}, \label{eq:mc_f_defn} \\
    f'(\av) &\defeq \frac{1}{2} \braket{a_1 | e^{i \beta_1 X} | a_2} \cdots \braket{ a_{p-1} | e^{i \beta_{p-1} X} | a_p} \braket{a_p | e^{i \beta_p X} | -a_0} \nonumber \\
    &\quad \times \braket{a_{0} | e^{-i \beta_p X} | a_{-p}} \braket{a_{-p} | e^{-i \beta_{p-1} X} | a_{-(p-1)}} \cdots \braket{a_{-2} | e^{-i \beta_1 X} | a_{-1}} \label{eq:mc_f_prime_defn}.
\end{align} 

Note that $\ket{-a_0}$ is equivalent to $X\ket{a_0}$, and that $f'$ is \emph{not} a derivative. With the form of \cref{eq:mc_ansatz_bitstring_sum} in hand, we may now search for an iterative formula to evaluate the expected energy of edge $(L,R)$.

Isolating a single leaf node $w$ of the subgraph, we can compute the sum over $\zv_w$

\begin{align}
    \sum_{\zv_w}  \exp\ofb{ -\frac{i}{\sqrt{D}}  \Gammav  \cdot (\zv_w \zv_{\parent(w)})}  f(\zv_w) \label{eq:H_defn_start},
\end{align}

where $\parent(w)$ denotes the unique parent of $w$ in the subgraph (as shown in \cite[Figure 2]{basso2022}). Exactly as in the case of \cite{basso2022}, all $D$ children of $\parent(w)$ contribute the same value the sum, so we define

\begin{align}
    H^{\of{1}}_D (\zv_{\parent(w)}) \defeq \of{\sum_{\zv_w}  \exp\ofb{ -\frac{i}{\sqrt{D}}  \Gammav  \cdot (\zv_w \zv_{\parent(w)})}  f(\zv_w)}^D,
\end{align}

to represent this contribution, which is equal for all parent nodes of all leaf nodes. We can then sum over parent nodes $\parent(w)$

\begin{align}
    \sum_{\zv_{\parent(w)}}  H^{(1)}_D(\zv_{\parent(w)}) \exp\ofb{- \frac{i}{\sqrt{D}} \Gammav  \cdot (\zv_{\parent(w)} \zv_{\parent(\parent(w))})} f(\zv_{\parent(w)}).
\end{align}

Once again, because $\parent(\parent(w))$ has $D$ identical children, which all contribute the same as $\parent(w)$, we have

\begin{align} 
    H_D^{(2)}(\zv_{\parent(\parent(w))}) &\defeq \of{\sum_{\zv_{\parent(w)}}  H^{(1)}_D(\zv_{\parent(w)}) \exp\ofb{- \frac{i}{\sqrt{D}} \Gammav  \cdot (\zv_{\parent(w)} \zv_{\parent(\parent(w))})} f(\zv_{\parent(w)})}^D.
\end{align}

Thus in general we have 

\begin{align}  
    H_D^{(m)}(\av) &\defeq \of{\sum_{\bv}  H^{(m-1)}_D(\bv) \exp\ofb{- \frac{i}{\sqrt{D}} \Gammav  \cdot (\av \bv)} f(\bv)}^D \nonumber \\
    &= \of{\sum_{\bv}  H^{(m-1)}_D(\bv) \cos\ofb{\frac{1}{\sqrt{D}} \Gammav  \cdot (\av \bv)} f(\bv)}^D , \\
    H_D^{(0)}(\av) &\defeq1 \label{eq:mc_H_defn},
\end{align}
for any $\av$ corresponding to a vertex that is a parent of the vertex corresponding to $\bv$. Note that in the last line, this is exactly the iteration in Eq.~4.11 of \cite{basso2022}, so we borrow that $H^{(m-1)}_D(\bv)=H^{(m-1)}_D(-\bv)$ to simplify the sum using the equality of the summand under the interchange $\bv \rightarrow -\bv$. We can repeat this process for $p$ iterations until we reach the level on the glued trees corresponding to the vertices where $\bv=L$ or $R$. Then it only remains to evaluate the following sum over $L$ and $R$:

\begin{align} 
    &\braket{\paramv | c_X X_L X_R +  c_Y Y_L Y_R + c_Z Z_L Z_R| \paramv} \nonumber \\
    &\hspace{120px}= \sum_{\av, \bv} \exp\ofb{-i\frac{\Gammav \cdot (\av \bv)}{\sqrt{D}}} H_D^{(p)}(\av) H_D^{(p)}(\bv)\nonumber \\
    &\hspace{150px}\times\of{c_X f'(\av)f'(\bv)- c_Y a_0b_0f'(\av)f'(\bv)+c_Z a_0b_0f(\av)f(\bv)},
\end{align} 

where $\av$, $\bv$ alias $\zv_L$, $\zv_R$, respectively, for compactness. Then, from \cite{basso2022} we have $H_D^{(p)}(-\av)=H_D^{(p)}(\av)$, and by \cref{eq:mc_f_defn} we can see that $f(-\av)=f(\av)$ and $f'(-\av)=f'(\av)$. Thus, the $XX$ term is even in $\av$, $\bv$ and the $YY$ and $ZZ$ terms are odd in $\av$, $\bv$, so we have

\begin{align}\begin{split}\label{eq:mc_xyz_exp_values}
    \braket{\paramv | X_L X_R | \paramv} &= \sum_{\av, \bv} f'(\av) f'(\bv) H^{(p)}_D(\av) H^{(p)}_D(\bv) \cos\left[\frac{\Gammav \cdot (\av \bv)}{\sqrt{D}}\right], \\
    \braket{\paramv | Y_L Y_R | \paramv} &= i\sum_{\av, \bv} a_0 b_0 f'(\av) f'(\bv) H^{(p)}_D(\av) H^{(p)}_D(\bv) \sin\left[\frac{\Gammav \cdot (\av \bv)}{\sqrt{D}}\right], \\
    \braket{\paramv | Z_L Z_R | \paramv} &= -i \sum_{\av, \bv} a_0 b_0 f(\av) f(\bv) H^{(p)}_D(\av) H^{(p)}_D(\bv) \sin\left[\frac{\Gammav \cdot (\av \bv)}{\sqrt{D}}\right]. 
\end{split}\end{align} 

As this computation only involves slight modifications \cite{basso2022}, the time and space complexity of evaluating this iteration (including all $H_D^{(m)}$ and $\nu_p(D, \gammav, \betav)$) remain $O(p 16^p)$ and $O(16^p)$, respectively. We numerically optimize parameters and report $\mathrm{MC}_{(p,\mathrm{QMC})}$, $\mathrm{MC}_{(p,\mathrm{XY})}$, and $\mathrm{MC}_{(p,\mathrm{EPR})}$ for small $p=\ofc{1,2,3,4,5}$ in \cref{sec:results/small_d}. 

\subsubsection{Infinite degree limit}\label{sec:iterations/mc_ansatz/d_inf}
We now prove \cref{cor:mc_inf_d_energy}. First, \cref{eq:mc_xyz_exp_values}, along with the definition of $\nu$ in \cref{eq:nu_DT} and \cref{eq:nu_H_DT}, allows us to compute $\nu_p(D,\mathrm{H},\mathrm{MC},\Thetav)$ for $\mathrm{H}=\ofc{\mathrm{H}_{\mathrm{QMC}}, \mathrm{H}_{\mathrm{XY}}, \mathrm{H}_{\mathrm{EPR}}}$. 

\begin{align}\label{eq:mc_qmc_nu_finite_d}
    &\nu_p(D, H, \mathrm{MC}, \Thetav) = \frac{\sqrt{D}}{2} \sum_{\av, \bv} H^{(p)}_D(\av) H^{(p)}_D(\bv) \\
    &\hspace{30px} \times \Bigg(\of{i\, a_0 b_0 \Big(c_Y f'(\av) f'(\bv) - c_Z f(\av) f(\bv)\Big) \sin\left[\frac{\Gammav \cdot (\av \bv)}{\sqrt{D}}\right]} +c_X f'(\av) f'(\bv)\cos\left[\frac{\Gammav \cdot (\av \bv)}{\sqrt{D}}\right]\Bigg),
\end{align}

where $\Thetav$ is specified by $(\gammav, \betav)$. We then apply the product rule of limits on $\sin$, $\cos$, and $H^{(p)}_D$ to find

\begin{align}
    \nu_p(H, MC, \gammav, \betav) &\defeq \lim_{D\rightarrow\infty} \nu_p(D, \gammav, \betav) \nonumber\\
    &= \frac{1}{2} \sum_{\av, \bv} H^{(p)}(\av) H^{(p)}(\bv)  \nonumber \\
    &\hspace{30px} \times \of{i\,a_0 b_0 \Big(c_Y f'(\av) f'(\bv) - c_Z f(\av) f(\bv)\Big)  \Gammav \cdot (\av \bv) +  \sqrt{D} c_X f'(\av) f'(\bv)} \label{eq:nu_mc_unequal_scaling},
\end{align}

where the limit $H^{(m)}\defeq\lim_{D\rightarrow \infty} H^{(m)}_D$ is shown in Section 4.2 of \cite{basso2022} to be

\begin{align}\label{eq:mc_h_lim_basso}
    H^{(m)}(\av) = \exp \ofb{ -\frac{1}{2} \sum_\bv f(\bv) H^{(m-1)}(\bv) \of{\Gammav \cdot (\av \bv) }^2 }.
\end{align}

We now prove argument 1) in the corollary. For the EPR Hamiltonian, the first two terms in \cref{eq:nu_mc_unequal_scaling} correspond to the contribution from $\braket{Z_L Z_R}$ and $\braket{Y_L Y_R}$, respectively. These terms scale as $\Theta(1)$ with respect to the degree. The third term, corresponding to the contribution from $\braket{X_L X_R}$, scales as $\Theta(\sqrt{D})$. Thus, the contributions from $Y$ and $Z$ components of the Hamiltonian vanish, leaving us to evaluate the $MC$ ansatz with respect to $1/2 \sum_{u,v} I_u I_v + X_u X_v$. This Hamiltonian is trivially maximized by the all $\ket{+}$ state, with energy $m$. This is precisely the energy obtained by the algorithm ZERO on EPR, so the algorithm is upper bounded by ZERO. The bound is tight, as this is the energy obtained by the MC ansatz with all angles set to $0$.

We then prove argument 2) in the corollary. Note that the contribution from the $\braket{X_L X_R}$ term can be expressed as

\begin{align*}
     \frac{c_X \sqrt{D}}{2} \sum_{\av, \bv} H^{(p)}(\av) H^{(p)}(\bv)  f'(\av) f'(\bv) = \frac{c_X \sqrt{D}}{2} \of{\sum_{\av} H^{(p)}(\av)  f'(\av)}^2.
\end{align*}

For QMC, $c_X < 0$, and we have from \cref{lemma:f_prime_h_sum_is_real} that $\sum_{\av} H^{(p)}(\av)  f'(\av)$ is real. Thus, the contribution from this term is always negative. As the normalized energy $\nu_p(QMC, MC, \gammav, \betav)$ for the MC ansatz on QMC is equal to the normalized energy $\nu_p(XY, MC, \gammav, \betav)$ of the MC ansatz on the XY model plus this term, we have that $\nu_p(QMC, MC, \gammav, \betav) \leq \nu_p(XY, MC, \gammav, \betav)$ for any angles.

Thus, for the rest of the proof, we turn our attention to the XY model, where $c_X=0$ and $c_Y=c_Z=-1$. We may then expand the dot product $\Gammav \cdot (\av \bv)$ to yield

\begin{align} 
    \nu_p(\gammav, \betav) 
    &= \frac{i}{2} \sum_{j=-p}^{p} \Gamma_j \Bigg(\of{\sum_{\av}       f(\av) H^{(p)}(\av)a_0 a_j}\of{\sum_{\bv} f(\bv)  H^{(p)}       (\bv)a_0 a_j} \nonumber\\
	&\hspace{53px} -\of{\sum_{\av} f'(\av) H^{(p)}(\av)a_0 a_j}\of{\sum_{\bv} f'(\bv) H^{(p)}(\bv)a_0 a_j}\Bigg), \\
	&= \frac{i}{2} \sum_{j=-p}^{p} \Gamma_j \of{\of{\sum_{\av} f(\av) H^{(p)}(\av)a_0 a_j}^2-\of{\sum_{\av} f'(\av) H^{(p)}(\av)a_0 a_j}^2}.  \label{eq:mc_nu_infinite_d}\\
\end{align}

We note the relations

\begin{align}
    &H^{(p)}(\av) = \exp\ofb{-\frac{1}{2}\sum_{j,k=-p}^p \Gamma_j \Gamma_k a_j a_k G_{j,k}^{(p-1)}} \label{eq:mc_ansatz_H_as_G}, \\
    &G_{j,k}^{(m)} \defeq \sum_{\av} f(\av) a_ja_k\exp\ofb{-\frac{1}{2}\sum_{j',k'} \Gamma_{j'} \Gamma_{k'} a_{j'} a_{k'} G_{j',k'}^{(m-1)}}, \\
    &G_{j,k}^{(0)} \defeq 1 \quad \forall j,k.
\end{align}

and define

\begin{align}
    G_{j,k}^{'(m)} &\defeq \sum_{\av} f'(\av) a_ja_k\exp\ofb{-\frac{1}{2}\sum_{j',k'} \Gamma_{j'} \Gamma_{k'} a_{j'} a_{k'} G_{j',k'}^{(m-1)}},
    \\
    G_{j,k}^{'(0)} &\defeq  1 \quad \forall (j,k),
\end{align}

from which it is clear that

\begin{align*} 
    \nu_p(\gammav, \betav) 
    &= \frac{i}{2} \sum_{j=-p}^{p} \Gamma_j  \ofb{\of{G_{0,j}^{(p)}}^2- \of{G_{0,j}^{'(p)}}^2}.
\end{align*}

The matrices $G_{j,k}^{(m)}$ for $1 \leq m \leq p$ were shown in Section A.4 of \cite{basso2022} to be iteratively computed (non-trivially!) in time and space $O(p^2 4^p)$ and $O(p^2)$, respectively. A simple modification to the iteration allows for the computation of $G^{'(p)}$ as well with constant multiplicative overhead. Using \cref{eq:mc_ansatz_H_as_G} we can evaluate $H^{(p)}(\av)$ for all $\av$ in time $O(p^2 4^p)$ and space $O(4^p)$, and the sum in \cref{eq:mc_nu_infinite_d} can then be evaluated in time $O(p^2 4^p)$. Given the bottlenecks, then, the total time and space requirements are $O(p^2 4^p)$ and $O(4^p)$ respectively. If we wish to save space, we can dynamically compute $H^{(p)}(\av)$ in \cref{eq:mc_nu_infinite_d} with \cref{eq:mc_ansatz_H_as_G}, which then reduces us back to $O(p^2)$ space but incurs $O(p^4 4^p)$ time. We find optimal parameters and report $\nu_p(\mathrm{XY}, \mathrm{MC})$ in \cref{sec:results/infinite_d}.

\subsection{XY ansatz}\label{apx:iteration_proofs/xy}
We now prove \cref{thm:xy_ansatz_energy}.  We proceed similarly to the MC ansatz case, by first determine the energy of a single edge $(L,R)$, again assuming a $(D+1)$-regular graph with girth $>2p+1$:
\begin{align}
    &\braket{\gammav_y, \gammav_z, \betav | (c_X X_L X_R + c_Y Y_L Y_R + c_Z Z_L Z_R) | \gammav_y, \gammav_z, \betav} \nonumber\\
    &\hspace{60px}=\bra{s} e^{i\gamma_{z_1} C_z} e^{i\gamma_{y_1} C_y} e^{i \beta_1 B} \cdots e^{i \gamma_{z_p} C_z} e^{i\gamma_{y_p} C_y} e^{i \beta_p B} (c_X X_L X_R + c_Y Y_L Y_R + c_Z Z_L Z_R) \nonumber\\
    &\hspace{75px}\times e^{-i \beta_p B} e^{-i\gamma_{y_p} C_y} e^{-i \gamma_{z_p} C_z} \cdots e^{-i \beta_1 B} e^{-i\gamma_{y_1} C_y} e^{-i\gamma_{z_1} C_z} \ket{s}.  
\end{align} 

This analysis follows closely to the case of the MC ansatz. However, we note that via the Heisenberg picture, the operators now depend on a subgraph of all vertices at most distance $2p$ from $L$ or $R$, due to the presence of \emph{two} entangling unitaries that apply to edges in the graph. Thus, now $n=2\of{D^{2p}+\cdots+D+1}$, and we again index the vertices by the set $\ofb{n}$. For simplicity, we first show the computation of $\braket{Z_L Z_R}$, and later show how to similarly compute $\braket{X_L X_R}$ and $\braket{Y_L Y_R}$. To evaluate $\braket{Z_L Z_R}$ we begin by inserting $2p+1$ identities via complete sets in the Pauli $Z$ basis and $2p$ identities via complete sets in the Pauli $Y$ basis
\begin{align}
    &\braket{\gammav_y, \gammav_z, \betav | Z_L Z_R | \gammav_y, \gammav_z, \betav} = \\
    &\hspace{12px}\sum_{\ofc{\zv^{[i]}}_{i\in\ofc{\mathcal{A}}}, \ofc{\yv^{[i]}}_{i\in\ofc{\mathcal{A}'}}} \hspace{-30px}
    \braket{s|\zv^{[1]}} \braket{\zv^{[1]}|e^{i \gamma_{z_1} C_z}|\zv^{[1]}} \braket{\zv^{[1]}|\yv^{[1]}} \braket{\yv^{[1]}|e^{i\gamma_{y_1}C_y}|\yv^{[1]}} \braket{\yv^{[1]} | e^{i \beta_1 B} |\zv^{[2]}} \times \cdots 
     \\ 
    &\hspace{50px}\times \braket{\zv^{[p]}|e^{i \gamma_{z_p} C_z}|\zv^{[p]}} \braket{\zv^{[p]}|\yv^{[p]}} \braket{\yv^{[p]}|e^{i\gamma_{y_p}C_y}|\yv^{[p]}} \braket{\yv^{[p]} | e^{i \beta_2 B} |\zv^{[0]}}  \braket{\zv^{[0]}|Z_L Z_R|\zv^{[0]}}  \\
    & \hspace{50px} \times \braket{\zv^{[0]} | e^{-i \beta_2 B} |\yv^{[-p]}} \braket{\yv^{[-p]}|e^{-i\gamma_{y_p}C_y}|\yv^{[-p]}} 
    \braket{\yv^{[-p]}|\zv^{[-p]}} \braket{\zv^{[-p]}|e^{-i \gamma_{z_p} C_z}|\zv^{[-p]}} \times  \cdots \\
    &\hspace{50px} \times \braket{\zv^{[-2]} | e^{-i \beta_1 B} |\yv^{[-1]}} \braket{\yv^{[-1]}|e^{-i\gamma_{y_1}C_y}|\yv^{[-1]}} 
    \braket{\yv^{[-1]}|\zv^{[-1]}} \braket{\zv^{[-1]}|e^{-i \gamma_{z_1} C_z}|\zv^{[-1]}} \braket{\zv^{[-1]}|s},
\end{align} 

where 
\begin{align}
    \mathcal{A} &= \ofc{1,2,\ldots,0,\ldots,-2,-1}, \\
    \mathcal{A}' &= \ofc{1,2,\ldots,p,-p,\ldots,-2,-1}. \label{eq:xy_a_set_defns}
\end{align}

We emphasize that $\ket{\yv^{[i]}}$ and $\ket{\zv^{[i]}}$ denote eigenstates in \emph{different bases}, hence terms such as $\braket{\yv^{[i]}|\zv^{[i]}}$ do not trivially enforce $\yv^{[i]}=\zv^{[i]}$. Identically to \cref{eq:mc_ansatz_bitstring_sum} of the MC ansatz case, we can replace the layer-by-layer sum over bit configurations $z^\ofb{i}_u, y^\ofb{i}_u$ by a vertex-by-vertex sum

\begin{align}
     \sum_{\{\zv_u\},\{\yv_u\}} z_L^{[0]} z_R^{[0]} \exp\Big[\frac{-i}{\sqrt{D}}\sum_{(u',v')\in E} \big( \Gammav_y \cdot (\yv_{u'} \yv_{v'}) + \Gammav_z \cdot (\zv_{u'} \zv_{v'}) \big)\Big] \prod_v h(\yv_{v}, \zv_{v}) g(\yv_{v}, \zv_{v})  \label{eq:xy_first_vector_form},
\end{align} 

where

\begin{align}
    &\{\yv_u\} \defeq (y_u^{[1]},\dots, y_u^{[p]}, y_u^{[-p]},\dots, y_u^{[-1]}), \quad \{\zv_u\} \defeq (z_u^{[1]}, \dots, z_u^{[p]}, z_u^{[0]}, z_u^{[-p]}, \dots, z_u^{[-1]}), \\
    &\Gammav_y \defeq (\gamma_1, \dots, \gamma_p, -\gamma_p, \dots, -\gamma_1), \quad \Gammav_z \defeq (\gamma_1, \dots, \gamma_p, 0, -\gamma_p, \dots, -\gamma_1), \\
    &h(\yv_u,\zv_u) \defeq \braket{z_u^{[1]}|y_u^{[1]}}\cdots\braket{z_u^{[p]}|y_u^{[p]}}\braket{y_i^{[-p]}|z_u^{[-p]}} \cdots \braket{y_u^{[-1]}|z_u^{[-1]}}, \\
    &\braket{y|z} \defeq \braket{z|y}^* = \frac{i^{(yz-y)/2}}{\sqrt{2}}, \\
    &g(\yv_u, \zv_u) \defeq \frac{1}{2} \braket{y_u^{[1]}|e^{i \beta_1 X_u}|z_u^{[2]}} \cdots \braket{y_u^{[p]}|e^{i \beta_p X_u}|z_u^{[0]}}\braket{z_u^{[0]}|e^{-i \beta_p X_u}|y_u^{[-p]}}\cdots\braket{z_u^{[-2]}|e^{-i \beta_1 X_u}|y_u^{[-1]}}, \\
    &\braket{y|e^{i \beta X}|z} \defeq \braket{z|e^{-i \beta X}|y}^*= \frac{i^{(yz-y)/2}}{\sqrt{2}}(\cos{\beta}+(-1)^{(yz-1)/2}  \sin{\beta}).\label{eq:xy_sum_definitions}\\
\end{align} 
To simplify notation, we rewrite the above sum. First, we index vectors in the following order:
\begin{align*}
    \av = \of{a_{z_1}, a_{y_1}, a_{z_2}, a_{y_2}, \dots a_{z_p}, a_{y_p}, 0, a_{y_{-p}}, a_{z_{-p}}, \dots, a_{y_{-2}}, a_{z_{-2}}, a_{y_{-1}}, a_{z_{-1}}}.
\end{align*}
Furthermore, we define
\begin{align}
    \Gammav \defeq \of{\Gamma_{z_1}, \Gamma_{y_1}, \Gamma_{z_2}, \Gamma_{y_2}, \dots, \Gamma_{z_p}, \Gamma_{y_p}, 0, -\Gamma_{y_p}, -\Gamma_{z_p}, \dots, -\Gamma_{y_2}, -\Gamma_{z_2}, -\Gamma_{y_1}, -\Gamma_{z_1}}.
\end{align}
Then we can rewrite \cref{eq:xy_first_vector_form} as
\begin{align}
     \braket{\gammav_y, \gammav_z, \betav | Z_L Z_R | \gammav_y, \gammav_z, \betav} =\sum_{\{\av_u\}} a_L^{[0]} a_R^{[0]} \exp\Big[\frac{-i}{\sqrt{D}}\sum_{(u',v')\in E} \Gammav \cdot (\av_{u'} \av_{v'}) \Big] \prod_v f(\av_v), \label{eq:xy_ansatz_vector_form}
\end{align} 
where
\begin{align}
    f(\av) = &\frac{1}{2} \braket{a_{z_1}|a_{y_1}} \braket{a_{y_1}|e^{i \beta_1 X_u}|a_{z_2}} \cdots \braket{a_{z_p}|a_{y_p}} \braket{a_{y_p}|e^{i \beta_p X_u}|a_{0}} \nonumber\\
    &\times \braket{a_{0}|e^{-i \beta_p X_u}|a_{y_{-p}}}\braket{a_{y_{-p}}|a_{z_{-p}}} \cdots \braket{a_{z_{-2}}|e^{-i \beta_1 X_u}|a_{y_{-1}}}\braket{a_{y_{-1}}|a_{z_{-1}}}\label{eq:xy_f_defn}.
\end{align} 

At this point, we note that we have a sum of the same form as the $\braket{Z_LZ_R}$ component of the MC ansatz \cref{eq:mc_ansatz_bitstring_sum}, albeit with different definitions of $\av$, $\Gammav$, $f$, etc. Indeed, we can introduce the corresponding functions

\begin{align}
    H^{(m)}_D(\av) &\defeq \of{\sum_{\bv} f(\bv) H^{(m-1)}_D(\bv) \exp\ofb{-\frac{i}{\sqrt{D}}\Gammav\cdot(\av \bv)}}^D, \\ 
    H^{(0)}_D(\av) &\defeq 1 \label{eq:xy_H_defn}.
\end{align}
As a consequence, if $f(\av)=f(-\av)$ and $H^{\ofc{m}}(\av)=H^{\ofc{m}}(-\av)$ for all $\av$ and $m$,
we may borrow \emph{all} of the analysis from \cref{eq:H_defn_start} to \cref{eq:mc_xyz_exp_values} to evaluate these sums.
We show these facts as \Cref{lemma:xy_f_of_negative_a_is_a} and \Cref{lemma:yz_H_negative_a_is_H_a} in \Cref{apx:usefullemmas}.
So we may borrow the analysis from before, noting that we need to sum over $2p$ layers. 
We arrive at the sum over the last two vertices $L$ and $R$
\begin{align} 
    \braket{\paramv | Z_L Z_R | \paramv} &= -i \sum_{\av, \bv} a_0 b_0 f(\av) f(\bv) H^{(2p)}_D(\av) H^{(2p)}_D(\bv) \sin\left[\frac{\Gammav \cdot (\av \bv)}{\sqrt{D}}\right] \label{eq:xy_zz_expectation_final},
\end{align} 
where similarly to the MC ansatz, $\av$ and $\bv$ alias $\zv_L$ and $\zv_R$ for conciseness.

The expectation values $\braket{Y_LY_R}$ and $\braket{Y_LY_R}$ can be evaluated via an almost identical argument. We start with  $\braket{Y_LY_R}$. We first use \cref{eq:inner_product_nonzero_cases} to see that the nonzero portion comes from the terms 

\begin{align}
    &\braket{\gammav_y, \gammav_z, \betav | Y_L Y_R | \gammav_y, \gammav_z, \betav} = \nonumber\\
    &\hspace{12px}\sum_{\ofc{\zv^{[i]}}_{i\in\ofc{\mathcal{A}}}, \ofc{\yv^{[i]}}_{i\in\ofc{\mathcal{A}'}}} \hspace{-30px}
    \braket{s|\zv^{[1]}} \braket{\zv^{[1]}|e^{i \gamma_{z_1} C_z}|\zv^{[1]}} \braket{\zv^{[1]}|\yv^{[1]}} \braket{\yv^{[1]}|e^{i\gamma_{y_1}C_y}|\yv^{[1]}} \braket{\yv^{[1]} | e^{i \beta_1 B} |\zv^{[2]}} \times \cdots 
     \nonumber\\ 
    &\hspace{50px}\times \braket{\zv^{[p]}|e^{i \gamma_{z_p} C_z}|\zv^{[p]}} \braket{\zv^{[p]}|\yv^{[p]}} \braket{\yv^{[p]}|e^{i\gamma_{y_p}C_y}|\yv^{[p]}} \braket{\yv^{[p]} | e^{i \beta_2 B} |{\zv^{[0]}}^{(L,R)}}  \braket{{\zv^{[0]}}^{(L,R)}|Y_L Y_R|\zv^{[0]}}  \nonumber\\
    & \hspace{50px} \times \braket{\zv^{[0]} | e^{-i \beta_2 B} |\yv^{[-p]}} \braket{\yv^{[-p]}|e^{-i\gamma_{y_p}C_y}|\yv^{[-p]}} 
    \braket{\yv^{[-p]}|\zv^{[-p]}} \braket{\zv^{[-p]}|e^{-i \gamma_{z_p} C_z}|\zv^{[-p]}} \times  \cdots \nonumber\\
    &\hspace{50px} \times \braket{\zv^{[-2]} | e^{-i \beta_1 B} |\yv^{[-1]}} \braket{\yv^{[-1]}|e^{-i\gamma_{y_1}C_y}|\yv^{[-1]}} 
    \braket{\yv^{[-1]}|\zv^{[-1]}} \braket{\zv^{[-1]}|e^{-i \gamma_{z_1} C_z}|\zv^{[-1]}} \braket{\zv^{[-1]}|s} \label{eq:xy_yy_expectation_start},
\end{align} 

where we use the same $\mathcal{A}$, $\mathcal{A}'$ from \cref{eq:xy_a_set_defns}. Much like in the MC ansatz case, we see that when recasting the sum into a sum of vertices, the only contributions that differ in the expectation of $\braket{Y_LY_R}$ from $\braket{Z_LZ_R}$ are those involving $\zv_L$ and $\zv_R$. In fact, we can borrow all the results from \cref{eq:mc_ansatz_before_bit_indexing} through \cref{eq:mc_xyz_exp_values} (under the present definitions of $\av$, $\Gammav$, etc) to arrive at 

\begin{align} 
    \braket{\paramv | Y_L Y_R | \paramv} &= i \sum_{\av, \bv} a_0 b_0 f'(\av) f'(\bv) H^{(2p)}_D(\av) H^{(2p)}_D(\bv) \sin\left[\frac{\Gammav \cdot (\av \bv)}{\sqrt{D}}\right] \label{eq:xy_yy_expectation_final},
\end{align} 

where

\begin{align}
    f'(\av) = &\braket{a_{z_1}|a_{y_1}} \braket{a_{y_1}|e^{i \beta_1 X_u}|a_{z_2}} \cdots \braket{a_{z_p}|a_{y_p}} \braket{a_{y_p}|e^{i \beta_p X_u}|-a_{0}} \nonumber \\
    &\times \braket{a_{0}|e^{-i \beta_p X_u}|a_{y_{-p}}}\braket{a_{y_{-p}}|a_{z_{-p}}} \cdots \braket{a_{z_{-2}}|e^{-i \beta_1 X_u}|a_{y_{-1}}}\braket{a_{y_{-1}}|a_{z_{-1}}} \label{eq:xy_f_prime_defn},
\end{align} 

and again where we use $f'(-\av)=f'(\av)$ in the derivation of \cref{eq:xy_yy_expectation_final} to arrive at its trigonometric form. For $\braket{X_L X_R}$, we again use \cref{eq:inner_product_nonzero_cases} to obtain an expression that is equal to \cref{eq:xy_yy_expectation_start}, albeit negated and without the presence of $a_0 b_0$. We can thus once again borrow all the results from \cref{eq:mc_ansatz_before_bit_indexing} through \cref{eq:mc_xyz_exp_values} (under the present definitions of $\av$, $\Gammav$, etc) to obtain 
\begin{align} 
    \braket{\paramv | X_L X_R | \paramv} &= \sum_{\av, \bv} a_0 b_0 f'(\av) f'(\bv) H^{(2p)}_D(\av) H^{(2p)}_D(\bv) \cos\left[\frac{\Gammav \cdot (\av \bv)}{\sqrt{D}}\right] \label{eq:xy_xx_expectation_final}.
\end{align} 

The complexity of evaluating $\nu_p(D, \gammav, \betav)$ is dominated by the iteration on $H^{(m)}_D(\av)$, as there are $O(2^{4p})$ values of $\av$ and $\bv$ that must be evaluated on at each step $1$ through $2p$, yielding time and space complexity of $O(p64^p)$ and $O(64^p)$, respectively.

\subsubsection{Infinite degree limit}\label{sec:iterations/xy_ansatz/d_inf}
The analysis of the infinite degree limit is identical to that of the MC ansatz in \cref{apx:iteration_proofs/mc}, under the substitution $p \rightarrow 2p$ and proper re-definitions of $\Gammav, \av$, etc. In particular, the argument from \cref{sec:iterations/mc_ansatz/d_inf} for statements 1) and 2) in the corollary are identical. Turning to the XY model then (where $c_X=0$, $c_Y=c_Z=-1$), the only logical step that remains to be shown is that the limit in \cref{eq:mc_h_lim_basso} exists and is finite. This limit is shown in \cite[Section 4.2]{basso2022} to follow from the fact \cite[Lemma 5]{basso2022} that for any $m$

\begin{align}
    \sum_{\av} f(\av) H_D^{(m)}(\av)=1.
\end{align}

We prove this (\Cref{lemma:yz_sum_fH_is_one}) in \Cref{apx:usefullemmas}. Thus, we obtain an identical final form

\begin{align} 
    \nu_p(\gammav, \betav) 
    &= \frac{i}{2} \sum_{j=-p}^{p} \Gamma_j  \ofb{\of{G_{0,j}^{(p)}}^2- \of{G_{0,j}^{'(p)}}^2} \label{eq:xy_nu_infinite_d},
\end{align}

where 
\begin{align}
    &H^{(p)}(\av) = \exp\ofb{-\frac{1}{2}\sum_{j,k=-p}^p \Gamma_j \Gamma_k a_j a_k G_{j,k}^{(p-1)}} \label{eq:xy_H_as_G}, \\
    &G_{j,k}^{(m)} \defeq \sum_{\av} f(\av) a_ja_k\exp\ofb{-\frac{1}{2}\sum_{j',k'} \Gamma_{j'} \Gamma_{k'} a_{j'} a_{k'} G_{j',k'}^{(m-1)}}, \\ 
    &G_{j,k}^{'(m)} \defeq \sum_{\av} f'(\av) a_ja_k\exp\ofb{-\frac{1}{2}\sum_{j',k'} \Gamma_{j'} \Gamma_{k'} a_{j'} a_{k'} G_{j',k'}^{(m-1)}}, \\
    &G_{j,k}^{(0)}, G_{j,k}^{'(0)} \defeq 1 \quad \forall j,k.
\end{align}

We now compute the time and space complexity of evaluating the iteration. We must first compute $O(p)$ matrices $G^{(m)}$, with $O(p^2)$ entries, each involving a sum over $O(p^2 16^p)$ terms, so in total the time and space requirements of computing the $G$ matrices are $O(p^5 16^p)$ and $O(p^2)$, respectively (we can compute the matrices in sequence from $1$ to $2p$). Then, to evaluate \cref{eq:xy_nu_infinite_d} we can either plug in \cref{eq:xy_H_as_G} in for $H^{(2p)}$, yielding a sum over $O(p)$ values $j$, where for each $j$, we compute two sums, each with $O(p^2 16^p)$ elements. This can be done in $O(1)$ space. Alternatively, one can pre-compute and store $H^{(2p)}(\av)$ for all $\av$ in time and space $O(p^2 16^p)$ and $O(16^p)$ and use this to compute \cref{eq:xy_nu_infinite_d} in time $O(p 16^p)$. The bottleneck still is in computing the $G$ matrices, so the time and space requirements can be taken to be $O(p^5 16^p)$ and $O(p^2)$.
\clearpage
\newpage

\section{Generalization to more expressive ans\"atze and larger locality}\label{apx:generic_ansatz}
We now show how to generalize the analysis presented in \cref{sec:iterations} for the MC and XY ans\"atze to allow for more multi-qubit gate interactions and larger locality of the target Hamiltonian. Specifically, we now aim to find high-energy states with respect to Hamiltonians of the form

\begin{align}
    H = \sum_{e \in E} H_e,
\end{align}

where the set $E$ denotes the hyperedges of a large girth\footnote{Like in \cite{basso2022}, girth is defined for hypergraphs as the minimum length of a Berge cycle in the graph, as defined in \cite{berge1989}.}, $(D+1)$-regular, $k$-uniform hypergraph, and $H_e$ denotes a Hamiltonian that acts non-trivially only on the qubits associated with the vertices in the hyperedge $e$ (i.e. it can be obtained by permuting qubits in the Hamiltonian $h \otimes I$, where $h$ is a Hamiltonian acting on exactly $k$ qubits). Furthermore, we require the $H_e$ is invariant under permutation of qubits associated with the vertices in $e$. This is analogous to the classical notion of a ``symmetric Boolean function''. We then observe that due to this symmetry we can express $H$ as

\begin{align}\label{eq:H_k_local_decomp}
    H = \sum_{e \in E} \sum_{\alpha=1}^{\binom{k+3}{3}} c_{\alpha} H^{\ofb{\alpha}}_e,
\end{align}

where the $c_{\alpha}$ are real valued coefficients and where each $H^{\ofb{\alpha}}_e$ denotes a Hamiltonian that is the tensor product of $n$ single-qubit Hamiltonians, where the Hamiltonian at qubit $i$ is some fixed $h^{\ofb{\alpha}}$ when $i \in e$ and $I$ when $i \notin e$, i.e $I^{\otimes e_1}\otimes h^{\ofb{\alpha}} \otimes I^{e_2-e_1} h^{\ofb{\alpha}} \otimes \cdots \otimes h^{\ofb{\alpha}} \otimes I^{k-e_k}$, where $e_j$ denotes the index of the $j$th vertex in $e$. In plain English, this means that a Hamiltonian consisting of permutation-invariant terms on $k$ qubits can be re-written as a linear combination of at most $O(k^3)$ terms that are simply of the form $h^{\otimes k}$, where $h$ is a single-qubit Hermitian matrix. This fact follows from Eq.11b of \cite{harrow2013} as well as Proposition 7.10 and Corollary 7.3 in \cite{watrous2018}, where $\textrm{dim}\of{\mathcal{X}}$ is the dimension of single qubit Hermitian matrices, which is $4$.\footnote{As an example, refer to $H_{QMC}$ as defined in \cref{eq:qmc_H}. The Hamiltonian consists of a sum over permutation invariant terms, and the Hamiltonian term corresponding to a single edge may be expressed via just four terms of the form $h\otimes h$, namely $h \in \ofc{I,X,Y,Z}$} Thus, to compute $\braket{H}$ it suffices to be able to compute each $\braket{H^{\ofb{\alpha}}_e}$, and there are only polynomially many such terms. 

Given that we allow for expanded locality in our Hamiltonian $H$, we can do the same for the multi-qubit gates in our ansatz. Accordingly, we allow ourselves to pick from the infinite set of gates of the form

\begin{align}
    C_j \defeq -\frac{1}{\sqrt{D}}\sum_{e \in E} P^{\ofb{j}}_e \label{eq:k_local_c_j_defn},
\end{align}

where $P^{\ofb{j}}_e$, is defined identically to $H^{\ofb{\alpha}}_e$ in \cref{eq:H_k_local_decomp} as the tensor product of the same single-qubit Hermitian $h^{\ofb{j}}$ of the qubits associated with the vertices in $e$. As an example, for the MC ansatz, we chose the single-qubit operator $h^{\ofb{j}}$ is $Z$, so that $P^{\ofb{0}}_e = Z_{e_1} Z_{e_2}$. For the XY ansatz, we have two such choices of $h^{\ofb{j}}$, namely $Y$ and $Z$. 

For the single qubit gates we likewise aim to allow for maximum generality, while still being in a form that we can analyze. Thus, we take our single-qubit ``mixing'' gates to be unitaries of the form $U(\beta)^{\otimes n}$, where $U$ is an arbitrary two-dimensional unitary rotation parameterized by the angle $\beta$. This definition naturally includes the $X$ mixer considered in the MC and XY ansatz. 

We are now ready to describe the class of ans\"atze that our method can analyze. In general the ansatz prepares the state 

\begin{align}
    \ket{\gammav, \betav} \defeq \of{\oleft{\prod_{\ell=1}^p} U\of{\beta_{\ell}}^{\otimes n} \of{\oleft{\prod_{j=1}^q} e^{-i \gamma_{j,\ell}C_j}}}\ket{\psi^{\otimes n}}\label{eq:k_local_ansatz_unitary_equation},
\end{align}

where the notation $\oright{\prod}$ denotes a product in order of smallest to largest index and  $\oleft{\prod}$, denotes a product in order of largest to smallest index. 

The ans\"atze in \cref{eq:k_local_ansatz_unitary_equation} consist of the following components: 

\begin{itemize}
    \item a symmetric initial product state of the form $\ket{\psi}^{\otimes n}$, where $\ket{\psi}$ is an arbitrary single qubit state,
    \item $p$ layers indexed by $\ell$, where each layer contains: 
    \subitem a) $q$ sub-layers of non-commuting $k$-qubit gates indexed by $j$ and parameterized $\gamma_{j,\ell}$,
    \subitem b) a single sub-layer of the same single qubit gate on all $n$ qubits, parameterized by $\beta_{\ell}$,
\end{itemize}

and where we defined $\gammav \defeq (\gamma_{1,1},\gamma_{2,1},\ldots,\gamma_{q,1},\gamma_{1,2},\gamma_{2,2},\ldots,\gamma_{q,\ell})$ and $\betav \defeq (\beta_1, \ldots, \beta_{\ell})$. This allows us to evaluate the expected energy of the ansatz on the Hamiltonian H of the form in \cref{eq:H_k_local_decomp} via

\begin{align}
    &\braket{\gammav, \betav| H |\gammav, \betav} = \sum_{e \in E} \sum_{\alpha} \braket{\gammav, \betav| H^{\ofb{\alpha}}_e |\gammav, \betav}\\
    &\hspace{20px}= \sum_{e \in E} \sum_{\alpha}\bra{\psi^{\otimes n}}\of{\oright{\prod_{\ell=1}^p}  \of{\oright{\prod_{j=1}^q} e^{i \gamma_{j,\ell}C_j}} U^{\dagger}\of{\beta_{\ell}}^{\otimes n}} H^{\ofb{\alpha}}_e \of{\oleft{\prod_{\ell=1}^p} U\of{\beta_{\ell}}^{\otimes n} \of{\oleft{\prod_{j=1}^q} e^{-i \gamma_{j,\ell}C_j}}}\ket{\psi^{\otimes n}}\label{eq:k_local_H_expectation_value_decomp}.
\end{align}

It may now be understood why we chose this form of an ansatz. Note that combined, the decomposition of $H$ into a linear combination of terms of the form $H^{\ofb{\alpha}}_e$ and the choice of gates $C^{\ofb{j}}$ allows us to recover the adiabatic theorem \cite{farhi2000a} for fixed $n$ in the limit of large depth $p$ and small angles. This can be seen in two ways. We first present a construction equivalent to the Hamiltonian Variational Ansatz (HVA) in \cite{wecker2015}. To see this, first observe that we may partition the Hamiltonian into $\sum_{\alpha} H^{\ofb{\alpha}}$, where $H^{\ofb{\alpha}}\defeq\sum_{e \in E} H^{\ofb{\alpha}}_e$. Then, if it is the case that for some $\alpha$, the maximal energy state of $H^{\ofb{\alpha}}$ is of form $\ket{\psi^{\otimes n}}$ for some $\ket{\psi}$, we can prepare our initial state as this $\ket{\psi^{\otimes n}}$. Then we can  choose our multi-qubit gates $C^{\ofb{j}}$ to be precisely the terms $H^{\ofb{\alpha}}$. Then, via the adiabatic theorem, if we evolve by the unitaries $C^{\ofb{j}}$ with small enough angles and for enough layers such that the Trotter product formula \cite{trotter1959} gives that we suitably approximate  evolution under $H$, we can recover the adiabatic theorem, without the need for mixing unitaries, i.e. $U(\beta)=I$. Indeed, this is the case for the EPR Hamiltonian, where $\ket{0}^{\otimes n}$ optimizes $H^{\ofb{\alpha}}=\sum_e Z_{e_1}Z_{e_2}$, and we can then evolve under terms such as $e^{i \gamma X_i X_j}$, $e^{i \gamma Y_i Y_j}$, and $e^{i \gamma Z_i Z_j}$ (we do not evolve by the identity as it simply accumulates a global phase). If we do not have a suitable symmetric product state that maximizes some portion of the Hamiltonian (e.g. for the QMC Hamiltonian), we can start of in the maximal-energy eigenstate of the mixing operator, which by definition is a symmetric product state, and by similar arguments we can slowly evolve into the maximal-energy eigenstate of $H$ given suitable time. However, we note that this simply serves as intuition for our ansatz, as we work in small, constant depth regimes. In this regime the first approach is more suitable, as we saw that for the QMC Hamiltonian, the lack of a good starting state makes improving over simple classical approaches with a constant number of layers prohibitive.

We may now then analyze a single term in \cref{eq:k_local_H_expectation_value_decomp}. We do this by once again inserting identities via complete basis sets whenever we encounter multi-qubit gates. This gives us

\begin{align}
     &\braket{\gammav, \betav| H^{\ofb{\alpha}}_e |\gammav, \betav} \nonumber\\
     &= \sum_{\ofc{\v^{\ofb{j,\ell}}}_{\substack{j\in \ofb{q}\\\ell\in\mathcal{A}'}}, \v^{\ofb{0}}} 
     \braket{\psi^{\otimes n}|\v^{\ofb{1,1}}} \nonumber\\
     &\times
     \Biggg(\oright{\prod_{\ell=1}^{p-1}}\Biggg(\oright{\prod_{j=1}^{q-1}} \braket{\v^{\ofb{j,\ell}}|e^{i\gamma_{j,\ell}C_j}|\v^{\ofb{j,\ell}}} \braket{\v^{\ofb{j,\ell}}|\v^{\ofb{j+1,\ell}}}\Biggg) \braket{\v^{\ofb{q,\ell}}|e^{i\gamma_{q,\ell}C_q}|\v^{\ofb{q,\ell}}} \braket{\v^{\ofb{q,\ell}}|U^{\dagger}\of{\beta_{\ell}}^{\otimes n}|\v^{\ofb{1,\ell+1}}}\Biggg)\nonumber\\
     &\times
     \Biggg(\oright{\prod_{j=1}^{q-1}} 
     \braket{\v^{\ofb{j,p}}|e^{i\gamma_{j,p}C_j}|\v^{\ofb{j,p}}}
     \braket{\v^{\ofb{j,p}}|\v^{\ofb{j+1,p}}} \Biggg) 
     \braket{\v^{\ofb{q,p}}|e^{i\gamma_{q,p}C_q}|\v^{\ofb{q,p}}} \braket{\v^{\ofb{q,p}}|U^{\dagger}\of{\beta_{p}}^{\otimes n}|\v^{\ofb{0}}} \nonumber\\
     &\times
     \braket{\v^{\ofb{0}}|H^{\ofb{\alpha}}_e|\v^{\ofb{0}}} \nonumber\\
     &\times
     \braket{\v^{\ofb{0}}|U\of{\beta_{p}}^{\otimes n}|\v^{\ofb{q,-p}}}
     \braket{\v^{\ofb{q,-p}}|e^{-i\gamma_{q,p}C_q}|\v^{\ofb{q,-p}}}
     \Biggg(\oleft{\prod_{j=1}^{q-1}} 
     \braket{\v^{\ofb{j+1,-p}}|\v^{\ofb{j,-p}}}
     \braket{\v^{\ofb{j,-p}}|e^{-i\gamma_{j,p}C_j}|\v^{\ofb{j,-p}}}
     \Biggg) \nonumber\\
     &\times
     \Biggg(\oleft{\prod_{\ell=1}^{p-1}}
     \braket{\v^{\ofb{0}}|U\of{\beta_{\ell}}^{\otimes n}|\v^{\ofb{q,-\ell}}}
     \braket{\v^{\ofb{q,-\ell}}|e^{-i\gamma_{q,\ell}C_q}|\v^{\ofb{q,-\ell}}}
     \Bigg(\oleft{\prod_{j=1}^{q-1}}
     \braket{\v^{\ofb{j+1,-\ell}}|\v^{\ofb{j,-\ell}}}
     \braket{\v^{\ofb{j,-\ell}}|e^{-i\gamma_{j,\ell}C_j}|\v^{\ofb{j,-\ell}}}\Bigg)\Biggg) \nonumber\\
     &\times\braket{\v^{\ofb{1,-1}}|\psi^{\otimes n}}\label{eq:k_local_insert_bases},
\end{align}

where $\mathcal{A}'=(1,2,\ldots,p,-p,\ldots,-2,-1)$ and where the $\ket{\v^{\ofb[j,\ell]}}$ are the $2^n$  eigenstates of $C_j$, and $\ket{\v^{\ofb{0}}}$ are the $2^n$ eigenstates of $H^{\ofb{\alpha}}_e$. For example, when evaluating $\braket{Z_L Z_R}$ in the MC ansatz, $C_j$ is diagonal in the Pauli $Z$ basis, so the eigenstates are simply $\ket{z}$ where $z\in\ofc{+1,-1}^n$. We then let $\v^{\ofb[j,\ell]}$, $\v^{\ofb{0}}$ be length-$n$, $\pm 1$-valued bitstrings that index the $2^n$ eigenstates. Via these definitions, we have that

\begin{align}
    \braket{\v^{\ofb{j,\ell}}|e^{i\gamma_{j,\ell}C_j}|\v^{\ofb{j,\ell}}} &= -\frac{i}{\sqrt{D}}\sum_{i \in e} \prod{i \in e}v^{\ofb{j,\ell}_i}, \\
    \braket{\v^{\ofb{0}}|H^{\ofb{\alpha}}_e|\v^{\ofb{0}}} &= \prod_{i\in e} v^{\ofb{0}_i}.
\end{align}

This allows us to, much like in the case of the MC or XY ansatz, replace the sum over the $2pq+1$ length-$n$ complete basis sets described by $\v^{\ofb[j,\ell]}$ and $\v^{\ofb{0}}$ and associated with each layer by an equivalent sum over $n$ length-$2pq+1$ complete basis sets described by $\v_i$.

\begin{align}
    \braket{\gammav, \betav| H^{\ofb{\alpha}}_e |\gammav, \betav}=\sum_{\v_i} f(\v_i)\exp\ofb{-\frac{i}{\sqrt{D}}\sum_{e' \in E} \Gammav \cdot \of{\prod_{i'\in e'} \v_{i'}}} \prod_{i\in e}v^{\ofb{0}}_i,
\end{align}

where we defined 

\begin{align}
    \Gammav &\defeq (\gamma_{1,1},\gamma_{2,1},\ldots,\gamma_{q,1},\gamma_{1,2},\gamma_{2,2}\ldots,\gamma_{q,\ell},0,\gamma_{q,-\ell},\ldots,\gamma_{2,-2},\gamma_{2,-1},\gamma_{q,-1},\ldots,\gamma_{2,-1},\gamma_{1,-1}), \label{eq:k_local_Gamma_defn}\\ 
    f(\v_i) &\defeq 
    \braket{\psi|v^{\ofb{1,1}}_i} 
    \Biggg(\oright{\prod_{\ell=1}^{p-1}}\Biggg(\oright{\prod_{j=1}^{q-1}} \braket{v^{\ofb{j,\ell}}_i|v^{\ofb{j+1,\ell}}_i}\Biggg) \braket{v^{\ofb{q,\ell}}_i|U^{\dagger}\of{\beta_{\ell}}|v^{\ofb{1,\ell+1}}_i}\Biggg) \nonumber\\
    &\times
    \Biggg(\oright{\prod_{j=1}^{q-1}} 
    \braket{v^{\ofb{j,p}}_i|\v^{\ofb{j+1,p}}_i} \Biggg) \braket{v^{\ofb{q,p}}_i|U^{\dagger}\of{\beta_{p}}|\v^{\ofb{0}}_i}
    \braket{v^{\ofb{0}}_i|U\of{\beta_{p}}|v^{\ofb{q,-p}}_i}
    \Biggg(\oleft{\prod_{j=1}^{q-1}} 
    \braket{v^{\ofb{j+1,-p}}_i|v^{\ofb{j,-p}}_i}
    \Biggg) \nonumber\\
    &\times
    \Biggg(\oleft{\prod_{\ell=1}^{p-1}}
    \braket{v^{\ofb{0}}_i|U\of{\beta_{\ell}}|v^{\ofb{q,-\ell}}_i}
    \Bigg(\oleft{\prod_{j=1}^{q-1}}
    \braket{v^{\ofb{j+1,-\ell}}_i|v^{\ofb{j,-\ell}}_i}
    \braket{v^{\ofb{1,-1}}_i|\psi}\Bigg)\Biggg) \label{eq:k_local_f_defn},
\end{align}

and $\prod_i \av_i$ denotes element-wise vector multiplication. We note that the computation of inner products such as $\braket{v^{\ofb{j,\ell}}_i|v^{\ofb{j',\ell'}}_i}$ can be computed in constant time, as they only rely on the two-dimensional Hilbert space of a single qubit. For each $C_j$, one may express the Hermitian unitary $h'^{\ofb{j}}$ obtained by rescaling the traceless component of the associated Hermitian $h^{\ofb{j}}$ described after \cref{eq:k_local_c_j_defn} in terms of its eigendecomposition $h'^{\ofb{j}}=Q_j\Lambda Q_j^{\dagger}$, where $\Lambda$ is equivalent to the Pauli $Z$ matrix, with $(1,-1)$ along its diagonal, and $Q_j$ is the $2\times2$ matrix, with its two columns given by the eigenvectors of $h'^{\ofb{j}}$. Then it can be easily seen that 

\begin{align}
    \braket{v^{\ofb{j,\ell}}_i|v^{\ofb{j,\ell'}}_i}=(Q_j^{\dagger}Q_{j'})_{v^{\ofb{j,\ell}}_i, v^{\ofb{j',\ell'}}_i},
\end{align}

where the subscript ${v^{\ofb{j,\ell}}_i, v^{\ofb{j',\ell'}}_i}$ denotes the row $v^{\ofb{j,\ell}}_i$ and column $v^{\ofb{j,\ell}}_i$, where 

\begin{align}
    Q_j^{\dagger}Q_{j'} = \begin{bmatrix}
                            (Q_j^{\dagger}Q_{j'})_{1,1} & (Q_j^{\dagger}Q_{j'})_{1,-1} \\
                            (Q_j^{\dagger}Q_{j'})_{-1,1} & (Q_j^{\dagger}Q_{j'})_{-1,-1}
                            \end{bmatrix}.
\end{align}

We note the structure of the subgraph in the reverse light-cone of $H^{\ofb{\alpha}}_e$ is exactly given by Fig 5 in \cite{basso2022}, i.e. it is a set of $k$ $D$-ary hypertrees glued about $e$, albeit in this case each hypertree is of depth $pq$ instead of $p$.\footnote{Note that in \cite{basso2022} $q$ refers to the arity of the cost function, here we take the more conventional choice of $k$ to be the arity, while $q$ refers to the number of unique multi-qubit gates considered in each layer} We can then proceed as in Sec. 8.4 of \cite{basso2022}, by summing over a set of leaf nodes $w_1,w_2,\ldots,w_{k-1}$ which all share a parent ($\parent(w_1)=\parent(w_2)=\cdots=\parent(w_k-1)$), which we denote as $\parent(w_1)$. Summing over the bit configurations $\v_{w_1},\v_{w_2},\ldots,\v_{w_k-1}$ then yields

\begin{align}
    \sum_{\v_{w_1},\ldots, \v_{w_{q-1}}} 
    \exp\ofb{ -\frac{i}{\sqrt{D}}  \Gammav \cdot (\v_{\parent(w_1)}\v_{w_1}\v_{w_2} \cdots \v_{w_{k-1}} )}\prod_{j=1}^{k-1}  f(\zv_{w_j}),
\end{align}

and because each of $w_1,w_2,\ldots,w_{k-1}$ contributes identically and is a function of $\parent(w_1)$ we obtain after summing over them

\begin{align}
    H_D^{(1)}(\v_{\parent(w_1)}) \defeq \of{\sum_{\v_{w_1},\ldots, \v_{w_{q-1}} }  
     \exp\ofb{-\frac{i}{\sqrt{D}} \Gammav \cdot (\v_{\parent(w_1)}\v_{w_1}\v_{w_2} \cdots \v_{w_{k-1}})} \prod_{j=1}^{k-1}  f(\v_{w_j})}^D.
\end{align}

We can then likewise sum over all $\v_{\parent(w_1)}$ sharing a parent $\parent{w'_1}\defeq\v_{\parent(\parent(w_1))}$. Let us label these vertices by $w'_1, w'_2, \ldots, w'_{k-1}$. The sums again are identical and a function of $\v_{\parent(\parent(w_1))}$, so we obtain the following function on $\v_{\parent(\parent(w_1))}$

\begin{align}
    H_D^{(2)}(\v_{\parent(w'_1)}) &\defeq
    \of{\sum_{\v_{w'_1},\ldots, \v_{w'_{k-1}}}   \exp\ofb{-\frac{i}{\sqrt{D}} \Gammav \cdot (\v_{\parent(w'_1)}\v_{w'_1}\v_{w'_2} \cdots \v_{w'_{k-1}})}\prod_{j=1}^{k-1}  \of{f(\v_{w'_j}) H_D^{(1)}(\v_{w'_j})}}^D,
\end{align}

we can repeat this iteration to define in general for a parent vertex $\av$ with children $\bv_1, \ldots \bv_{k-1}$

\begin{align}
    H_D^{(m)}(\v_{a}) &\defeq
    \of{\sum_{\bv_1,\ldots, \bv_{k-1}}  \exp\ofb{-\frac{i}{\sqrt{D}} \Gammav \cdot (\av \bv_1 \cdots \bv_{k-1})}\prod_{j=1}^{k-1}  \of{f(\bv_i) H_D^{(m-1)}(\bv_i)}}^D.
\end{align}

And we can then finally evaluate $\braket{\gammav, \betav| H^{\ofb{\alpha}}_e |\gammav, \betav}$ via

\begin{align}
    \braket{\gammav, \betav| H^{\ofb{\alpha}}_e |\gammav, \betav} = \sum_{\ofc{\v_i}_{i\in e}} \of{\prod_{i\in e} v^{\ofb{0}}_i} \exp\ofb{-\frac{i}{\sqrt{D}} \Gammav \cdot (\v_{e_1} \v_{e_2} \cdots \v_{e_k})} \prod_{j \in e} f(\v_j)H_D^{(pq)}(\v_j).
\end{align}

We note that the complexity of the iteration is identical to that of \cite{basso2022}, albeit with $p \times q$ effective layers. Thus we can borrow their analysis of the time and memory complexity to find that the iteration has an overall time complexity of $\mathcal{O}(pq 4^{pqk})$ and memory complexity of $\mathcal{O}(4^{pq})$. We do not presently show how to take the $D \rightarrow \infty$ limit, though it can be done by following methods similar to \cref{sec:iterations/mc_ansatz/d_inf} and \cref{sec:iterations/xy_ansatz/d_inf}.

\clearpage
\newpage

\section{Optimal variational parameters on the XY Hamiltonian in the infinite degree limit}
\label{apx:angles}
As discussed in \cref{sec:iterations/mc_ansatz/d_inf} and \cref{sec:iterations/xy_ansatz/d_inf}, in the infinite-degree limit, the MC and XY ansatz yield nontrivial energies only for the XY Hamiltonian. We now present the optimal angles for the two ans\"atze for this Hamiltonian, obtained as described in \cref{sec:results}. Following \cref{rem:mixer_commute}, in the following tables the final angle $\beta_p$ is set to $0$.

\begin{table}[ht]
    \centering
    \begin{tabular}{|l|l|}
        \hline
        $p$ & $\gammav$ \\
        \hline
        1 & [0.0] \\
        \hline
        2 & [0.5243, 0.7434] \\
        \hline
        3 & [0.3892, 0.7070, 0.5182] \\
        \hline
        4 & [0.3297, 0.5688, 0.6406, 2.1246] \\
        \hline
        5 & [0.2950, 0.5144, 0.5586, 0.6429, 2.2040] \\
        \hline
        6 & [0.2705, 0.4804, 0.5074, 0.5646, 0.6396, 2.4378] \\
        \hline
        7 & [0.2853, 0.4998, 0.5391, 0.6157, 0.8918, 0.1616, 0.2659] \\
        \hline
        8 & [0.2640, 0.4720, 0.4953, 0.5473, 0.6167, 0.9041, 0.1569, 0.2631] \\
        \hline
        9 & [0.2483, 0.4469, 0.4676, 0.5039, 0.5489, 0.6201, 0.9124, 0.1530, 0.2606] \\
        \hline
        10 & [0.2348, 0.4282, 0.4468, 0.4764, 0.5047, 0.5542, 0.6237, 0.9203, 0.1499, 0.2581] \\
        \hline
    \end{tabular}
    
    \vspace{0.3cm} % Add vertical space between the two parts
    
    \begin{tabular}{|l|l|}
        \hline
        $p$ & $\betav$ \\
        \hline
        1 & [0.0] \\
        \hline
        2 & [0.3997, 0.0] \\
        \hline
        3 & [0.4714, 0.2811, 0.0] \\
        \hline
        4 & [0.5500, 0.3675, 0.2109, 0.0] \\
        \hline
        5 & [0.5710, 0.4176, 0.3028, 0.1729, 0.0] \\
        \hline
        6 & [0.5899, 0.4492, 0.3559, 0.2643, 0.1485, 0.0] \\
        \hline
        7 & [0.5753, 0.4272, 0.3193, 0.2056, 0.6244, -0.6133, 0.0] \\
        \hline
        8 & [0.5930, 0.4553, 0.3656, 0.2809, 0.1769, 0.6386, -0.6364, 0.0] \\
        \hline
        9 & [0.6018, 0.4696, 0.3940, 0.3270, 0.2475, 0.1533, 0.6521, -0.6548, 0.0] \\
        \hline
        10 & [0.6091, 0.4821, 0.4119, 0.3595, 0.2947, 0.2218, 0.1357, 0.6628, -0.6685, 0.0] \\
        \hline
    \end{tabular}
    \caption{\footnotesize Optimal angles for the MC ansatz on the XY Hamiltonian for various depths $p$. Angles are normalized using time-reversal symmetry of the QAOA (flipping signs of all angles) to be positive. The final values $\beta_p$ are set to $0.0$ as they do not affect the energy.}
    \label{tab:angles_mc}
\end{table}

\begin{table}[ht]
    \centering
    \begin{tabular}{|l | l | l| l|}
        \hline
        $p$ & $\gammav_z$ & $\gammav_y$ & $\betav$ \\
        \hline
        1 & [0.0] & [0.0] & [0.0] \\
        \hline
        2 & [-0.0940, -1.0332] & [-0.6813, 0.3340] & [0.3667, -0.0] \\
        \hline
        3 & [-0.1024, -0.4958, -0.9400] & [-0.5696, 0.3408, 0.1790] & [-1.2382, -0.1569, 0.0] \\
        \hline
        4 & [0.1217, -0.1201, 0.1346, 0.0] & [-0.2930, 0.6766, 0.9354, 0.3066] & [-0.2865, 0.2786, 0.1117, -0.0] \\
        \hline
    \end{tabular}
    \caption{\footnotesize Optimal angles for $\nu_p$ with the XY ansatz on the XY Hamiltonian up to depth $p=4$. Final values $\beta_p$ are set to $0.0$ as they do not affect the energy.}
    \label{tab:angles_xy}
\end{table}

\end{document}